\documentclass[12pt]{article}

\newtheorem{theorem}            {Theorem}[section]
\newtheorem{corollary}          [theorem]{Corollary}

\newtheorem{sidexample}            [theorem]{Example}
\newtheorem{lemma}              [theorem]{Lemma}
\newtheorem{sideremark}         [theorem]{Remark}
\newtheorem{sidenote}           [theorem]{Note}

\newcommand{\qed}               {\hfill $\Box$}

\usepackage{amsmath} 
\usepackage{amssymb}  
\usepackage{amsfonts}
\usepackage{mathrsfs}

\usepackage{graphicx}

\title{Quantum Filtering for Systems Driven by Fields in Single Photon States and Superposition of Coherent States
using Non-Markovian Embeddings\thanks{This work was
supported by the Australian Research Council and the UK Engineering and Physical Sciences Research Council grant EP/G039275/1.}}

\author{John E.~Gough\thanks{Institute for Mathematics and Physics, Aberystwyth University, SY23 3BZ, Wales, United Kingdom. Email: jug@aber.ac.uk }
\and Matthew R.~James\thanks{ARC Centre for Quantum Computation and Communication Technology, Research School of Engineering, Australian National University, Canberra, ACT 0200, Australia. Email: Matthew.James@anu.edu.au}\and
Hendra I.~Nurdin\thanks{Research School of Engineering, Australian National University, Canberra, ACT 0200, Australia. Email: Hendra.Nurdin@anu.edu.au}}

\begin{document}
\maketitle

\begin{abstract}
The purpose of this paper is to determine quantum master and filter equations for systems coupled to fields in certain non-classical continuous-mode states. Specifically, we consider two types of field states (i) single photon states, and (ii) superpositions of coherent states.
The system and field are described using a quantum stochastic unitary model. Master equations are derived from this model and are given in terms of  systems  of coupled equations.
The output field carries information about the system, and is continuously monitored.
The quantum filters  are  determined with the aid of an embedding of the system into a larger non-Markovian system, and are  given by a system of coupled stochastic differential equations.
\end{abstract}


\pagestyle{myheadings}
\thispagestyle{plain}
\markboth{Gough, James, Nurdin}{Single Photon Quantum Filter}

\section{Introduction}

In recent years single photon states of light and superpositions of coherent states
have become increasingly
important due to applications in quantum technology, in particular, quantum
computing and quantum information systems, \cite{MHNWGBRH97}, \cite{GM08},
\cite{KLM01}, \cite{GRTZ02}, \cite{VWSRVSKW06}. For instance, the light may
interact with a system, say an atom, quantum dot, or cavity, and this system may be used
as a quantum memory, \cite{MHNWGBRH97}, or  to control
the pulse shape of the single photon state \cite{GM08}. When light interacts
with a quantum system, information about the  system
is contained in the scattered light. This information may be useful for
monitoring the behavior of the  system, or for controlling it. The
topic of this paper concerns the extraction of information from the
scattered light when the incoming light is placed in a single photon state $\vert \Psi \rangle = \vert 1_\xi \rangle$, or a superposition of coherent states $\vert \Psi \rangle = \sum_j \alpha_j \vert f_j \rangle$, as illustrated in Figure \ref{fig:filter-one-1}.

\begin{figure}[h]
\begin{center}
\includegraphics{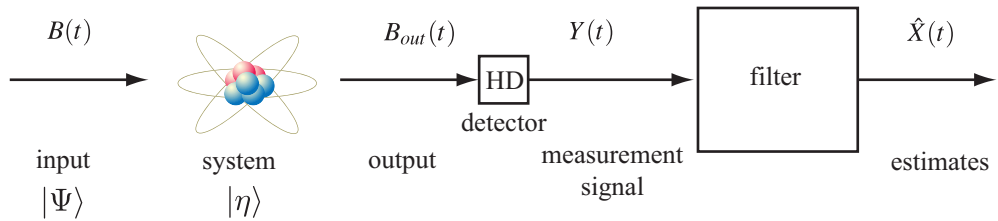} 
\end{center}
\caption{A system initialized  in a state $\vert \eta \rangle$ coupled to a field in a state $\vert \Psi  \rangle$ (single photon or superposition of coherent states). The output
field is continuously monitored by homodyne detection (assumed perfect) to produce a classical
measurement signal $Y(t)$. The output $Y(t)$ is filtered 
 to produce estimates $\hat X(t)=\pi_t(X)$ of system operators $X$ at time $t$.
}
\label{fig:filter-one-1}
\end{figure}

The problem of extracting information from continuous measurement of the
scattered light is a problem of \emph{quantum filtering}, \cite{BB91}, \cite
{VPB92}, \cite{VPB92a}, \cite{HC93}, \cite{WM93}, \cite{AB03},  \cite{BHJ07}, \cite{WM10}.
The current state of the art for quantum filtering considers incoming light
in a vacuum or other Gaussian state, with quadrature or counting measurements.
Both single photon states of light, and superpositions of coherent states of light,  are highly non-classical, and are
fundamentally different from Gaussian  states.
 In view of the
increasing importance of these non-Gaussian
states of light, the purpose of this
paper is to solve a quantum filtering problem for systems driven by fields in single
photon states  and superpositions of  coherent states.

 In the case of single photon fields, the master equation describing unconditional dynamics was shown to be a system of coupled equations in \cite{GEPZ98}, a feature of  non-Markovian character.  Markovian embeddings were used in \cite{HPB04} to
derive quantum trajectory equations (quantum filtering equations) for a class of non-Markovian master equations. In recent work, the authors have shown how to construct ancilla systems to combine with the system of interest to form a Markovian extended system driven by vacuum from which quantum filtering results may be obtained for single photon states and superpositions of coherent states from the standard filter for the extended system, \cite{GJN11a}, \cite{GJNC11}. However, depending on the complexity of the non-classical state, it may possibly be difficult to determine suitable ancilla systems, and indeed the superposition case was not straightforward.   In this paper we present an alternative approach to the embedding that also allows for the derivation of the quantum filter. 
The extended system forms a non-Markovian system, with the ancilla, system and field initialized in a superposition state. While standard filtering results do not apply, the quantum stochastic methods can  nevertheless be applied to determine the quantum filters. In this way, we expand the range of methods that may be applied to derive quantum filters for non-classical states.

The paper is organized as follows. In section \ref{sec:problem} the idealized 
filtering problem to be solved in this paper is formulated. 
The continuous mode single photon states are defined in Section \ref{sec:photon-state}.
 The master equation for the single photon field state is
derived in Section \ref{sec:master} using the model presented in Section \ref
{sec:problem}. This leads naturally to Section \ref{sec:the-matrix}, where
the system is embedded in a larger  model, inspired by the approach
used in \cite{HPB04} but differing in the details. This extended system provides a compact and
transparent description of the problem, and may readily be generalized to $n$%
-photon states, and indeed, multiple channels of $n$-photon states. The
quantum filter for the extended system is presented in Section \ref{sec:matrix-filter},  with a  derivation extending the reference
method appearing in Section \ref{sec:extended-filter}. 
The filtering results for the extended system are used to find the filtering equations for the
original problem involving a single photon field in Section \ref{sec:single-photon-filter}.
The superposition of coherent field states is defined in Section \ref{sec:cat-super}, and a suitable embedded system for this case is described in Section \ref{sec:cat-embed}. The corresponding master and filtering equations are presented in Sections \ref{sec:cat-master} and \ref{sec:cat-super-filter}, respectively. 
Some concluding remarks are made in Section \ref{sec:conclusion}.  

In this paper we are not concerned with technical issues concerning domains of unbounded operators and related matters, and indeed, we assume that the system operators are bounded, and that all quantum stochastic integrals are well-defined in the sense of Hudson-Parthasarathy, \cite{HP84}.

\emph{Notation:} We use the standard Dirac notation $\vert \psi \rangle$ to
denote state vectors (vectors in a Hilbert space) \cite{EM98}, \cite{AFP09}.
The superscript $^\ast$ indicates Hilbert space adjoint or complex
conjugate. The inner product of state vectors $\vert \psi_1 \rangle $ and $\vert \psi_2 \rangle$  is denoted $\langle \psi_1 \vert
\psi_2 \rangle$. The expected value of an operator $X$ when the system is in state $\vert \psi \rangle$ is denoted $\mathbb{E}_\psi[ X ] = \langle \psi \vert X \vert \psi \rangle$.
For operators $A$ and $B$ we write
$
\langle A, B \rangle = \mathrm{tr}[ A^\ast B ].
$

\section{Problem Formulation}
\label{sec:problem}

We consider a quantum system $S$ coupled to a quantum field $B$, as shown in
Figure \ref{fig:filter-one-1}. The field $B$ has two components, the input
field $B_{in}$ and the output field, $B_{out}$. In this paper we consider two non-classical cases for the state $\vert \Psi \rangle$ of the input field (i) 
a single photon state  $\vert \Psi \rangle = \vert 1_\xi \rangle$, where $\xi$ is a complex valued function such that $\int_0^\infty \vert
\xi(s) \vert^2 ds =1$ (representing the wave packet shape), or  (ii) a superposition of coherent states $\vert \Psi \rangle  = \sum_j \alpha_j \vert f_j \rangle$,
where $\vert f_j \rangle$ are coherent states and the complex numbers $\alpha_j$  ($j=1,\ldots,n$) are normalized weights. 

As illustrated in Figure \ref{fig:filter-one-1}, the field 
interacts with the quantum system $S$, and
the results of this interaction provide information about the system that
may be obtained through continuous measurement of an observable $Y(t)$ of
the output field $B_{out}(t)$. The filtering problem of interest in this
paper is to determine the conditional state from which estimates $\hat X(t)$
of system operators $X$ may be determined at time $t$ based on knowledge of
the observables $\{ Y(s)$, $0 \leq s \leq t \}$.

In what follows the system $S$ is assumed to be defined on a Hilbert space $%
\mathfrak{H}_S$, with an initial state denoted $\vert  \eta \rangle   \in %
\mathfrak{H}_S$. The input field $B_{in}$ is described in terms of
annihilation $B(\xi)$ and creation $B^\ast(\xi)$ operators defined on a Fock
space $\mathfrak{F}$, \cite[Chapter II]{KRP92}, \cite[Section 4]{BHJ07}.
Quantum expectation will be denoted by the symbol $\mathbb{E}$, and when we wish to display the underlying state, we employ subscripts;  for example, $\mathbb{E}_{\eta\Psi}$ denotes quantum expectation with respect to the state $\vert \eta \rangle \otimes \vert \Psi \rangle$.

The dynamics of the system will be described using the quantum stochastic
calculus, \cite{HP84}, \cite{GC85}, \cite{KRP92}, \cite{GZ00}, \cite{BHJ07}.
Quantum stochastic integrals are
defined in terms of fundamental field operators $B(t)$, $B^\ast(t)$ and $%
\Lambda(t)$, \cite[Chapter II]{KRP92}, \cite[Section 4]{BHJ07}.\footnote{%
In terms of annihilation and creation white noise operators $b(t), b^\ast(t)$
that satisfy singular commutation relations $[b(s), b^\ast(t)]=\delta(t-s)$,
the fundamental field operators are given by $B(t) = \int_0^t b(s) ds$, $%
B^\ast(t)= \int_0^t b^\ast(s) ds$, and $\Lambda(t)= \int_0^t b^\ast(s) b(s)
ds$. Also, we may write $B(\xi) = \int_0^\infty \xi^\ast(s) dB(s)$.} The
non-zero Ito products for the field operators are
\begin{equation}
dB(t) dB^\ast(t) = dt, \ \ dB(t) d\Lambda(t) = dB(t), \ \ d\Lambda(t)
d\Lambda(t) = d\Lambda(t), \ \ d\Lambda(t) dB^\ast(t)=dB^\ast(t). \label{eq:Ito-sp}
\end{equation}

The dynamics of the composite system is described by a unitary $U(t)$
solving the Schr\"{o}dinger equation, or quantum stochastic differential equation (QSDE),
\begin{equation}
dU(t) = \{ (S-I)d\Lambda(t) + L dB^\ast (t)- L^\ast S dB(t) - (\frac{1}{2}
L^\ast L +iH )dt \} U(t),  \label{eq:unitary}
\end{equation}
with initial condition $U(0)=I$. Here, $H$ is a fixed self-adjoint operator
representing the free Hamiltonian of the system, and $L$ and $S$ are system
operators determining the coupling of the system to the field, with $S$
unitary. In this paper, for simplicity we assume that the parameters $S,L,H$ are bounded operators
on the system Hilbert space $\mathfrak{H}_S$. However, we remark under some suitable additional conditions the results
and equations obtained in this paper should also be extendable to some special classes of QSDEs with unbounded parameters,
exploiting the results in \cite{Fagno90,FW03}.

A system operator $X$ at time $t$ is given in the Heisenberg picture by $%
X(t)=j_{t}( X) =U( t) ^{\ast } ( X\otimes I ) U( t) $ and it follows from the
quantum Ito calculus that
\begin{eqnarray}
dj_{t}( X) &=&j_{t}(S^{\ast }XS-X) d\Lambda ( t) +j_{t}(S^{\ast }[X,L]) dB(t) ^{\ast }  
\notag \\
&& 
+j_{t}([L^{\ast },X]S) dB( t) +j_{t}(\mathcal{L}(X)) dt ,
\label{eq:X-dyn}
\end{eqnarray}
where
\begin{equation}
\mathcal{L}(X)=\frac{1}{2}L^{\ast }[X,L]+\frac{1}{2}[L^{\ast },X]L-i[ X,H] .
\end{equation}
The map $X\mapsto \mathcal{L} (X)$ is known as the \emph{Lindblad generator},
while the quartet of maps $X \mapsto \mathcal{L} (X), S^\ast XS -X, \,
S^\ast [X,L ], \, [L^\ast , X] S$ are known as \emph{Evans-Hudson maps}.

The output field is defined by $B_{out}(t) = U(t)^\ast B(t) U(t)$.\footnote{%
Recall $B(t)=B_{in}(t)$ is the input field.} In this paper we consider the
output field observable $Y(t)$ defined by
\begin{equation}
Y(t) = U(t)^\ast Z(t) U(t) , 
 \label{eq:Y-out}
\end{equation}
where
\begin{equation}
Z(t)= B(t)+ B^\ast(t), 
 \label{eq:Z-def}
\end{equation}
is a quadrature  observable of the input field (the counting case $Z(t)=\Lambda(t)$ is discussed briefly in Section \ref{sec:conclusion}).  Note that both $Z(t)$ and $Y(t)$ are
self-adjoint and self-commutative: $[Z(t), Z(s)]=0$ and $[Y(t), Y(s)]=0$. We
write $\mathscr{Z}_t$ and $\mathscr{Y}_t$ for the subspaces of commuting
operators generated by the observables $Z(s)$, $Y(s)$, $0\leq s \leq t$,
respectively.\footnote{$\mathscr{Z}_t$ and $\mathscr{Y}_t$ are commutative
von Neumann algebras. They are also filtrations, e.g. $\mathscr{Z}_{t_1}
\subset \mathscr{Z}_{t_2}$ whenever $t_1 <t_2$.} They are related by the
unitary rotation $\mathscr{Y}_t = U(t)^\ast \mathscr{Z}_t U(t)$. Physically,
$Y(t)$ may represent the integrated photocurrent arising in an idealized (perfect)  homodyne
photodetection scheme, as in Figure \ref{fig:filter-one-1}. For further information on homodyne detection, we refer the reader to the literature;  for example, \cite{BR04},  \cite{AB03}, 
\cite{WM10}.

The primary goal of this paper is to determine the \emph{quantum filter} for
the quantum conditional expectation (see, e.g. \cite[Definition 3.13]{BHJ07}%
)
\begin{equation}
\hat X(t) = \mathbb{E}_{\eta\Psi}[ X(t) \, \vert \, \mathscr{Y}_t ] .
\label{eq:cond-exp}
\end{equation}
This conditional expectation is well defined, since $X(t)$ commutes with the
subspace $\mathscr{Y}_t $ (non-demolition condition). The conditional
estimate $\hat X(t)$ is affiliated to $\mathscr{Y}_t$ (written in abbreviated fashion as  $\hat X(t) \in \mathscr{Y}_t$)
and is characterized by the requirement that
\begin{equation}
\mathbb{E}_{\eta\Psi} [ \hat X(t) K ] = \mathbb{E}_{\eta\Psi} [ X(t) K ]
\label{eq:c-exp-def}
\end{equation}
for all $K \in \mathscr{Y}_t$.

\section{Single Photon Input Fields}
\label{sec:photon}

\subsection{Single Photon  Fields States}
\label{sec:photon-state}

In this section we consider the  continuous-mode single photon state $\vert \Psi \rangle = \vert 1_\xi \rangle$ 
 defined by \cite[sec. 6.3]{RL00}, \cite[eq. (9)]{GM08}
\begin{equation}
\vert 1_\xi \rangle = B^\ast(\xi) \vert 0  \rangle,  
\label{eq:xi-create}
\end{equation}
 where $\xi$ is a complex valued function such that $\int_0^\infty \vert
\xi(s) \vert^2 ds =1$, and 
$\vert 0 \rangle$ is the vacuum state of the field. Expression (\ref{eq:xi-create}) says that  the single photon wavepacket with temporal shape $\xi$ is created from the vacuum using the field operator $B^\ast(\xi)$.

The Hilbert space for the composite system is
\begin{equation*}
\mathfrak{H} = \mathfrak{H}_S \otimes \mathfrak{F}= \mathfrak{H}_S \otimes %
\mathfrak{F}_{t]} \otimes \mathfrak{F}_{(t},
\end{equation*}
where here we have exhibited the continuous temporal tensor product
decomposition of the Fock space $\mathfrak{F}=\mathfrak{F}_{t]} \otimes %
\mathfrak{F}_{(t}$ into past and future components, which is of basic
importance in what follows.  Write
\begin{equation}
\mathbb{E}_{11} [ X \otimes F ] = \langle \eta 1_\xi \vert  (X \otimes F) \vert \eta
1_\xi \rangle = \langle \eta  \vert X \vert \eta \rangle \langle 1_\xi \vert  F  \vert 1_\xi \rangle
\end{equation}
for the expectation with respect to the product state $\vert \eta 1_\xi \rangle$, where the field is in the single
photon state.  Here and in what follows  $X$ is a bounded system operator acting on $\mathfrak{H}_S$,
and $F$ is a field operator acting on the Fock space $\mathfrak{F}$.
Similarly, we may define the expectation when the field is in the vacuum
state,
\begin{equation}
\mathbb{E}_{00} [ X \otimes F ] = \langle  \eta 0 \vert  (X \otimes F) \vert \eta
0  \rangle = \langle \eta \vert X \vert \eta \rangle \langle 0 \vert  F \vert 0  \rangle .
\end{equation}
We will also have need for the cross-expectations
\begin{eqnarray}
\mathbb{E}_{10}[ X \otimes F ] = \langle \eta 1_\xi \vert  (X \otimes F)  \vert \eta 0
\rangle, \ \text{and} \ \mathbb{E}_{01}[ X \otimes F ] = \langle \eta 0  \vert
(X \otimes F) \vert  \eta 1_\xi \rangle.
\end{eqnarray}

A crucial difference between the single photon state and the vacuum state is
that the later state factorizes $\vert 0 \rangle =\vert 0_{t]} \rangle
\otimes \vert 0_{(t} \rangle $ with respect to the temporal factorization
$\mathfrak{F}=\mathfrak{F}_{t]} \otimes \mathfrak{F}_{(t}$ of the Fock
space, with $\vert 0_{t]} \rangle \in \mathfrak{F}_{t]} $ and $\vert
0_{(t} \rangle \in \mathfrak{F}_{(t}$, while the former does not. Rather,
we have
\begin{equation}
\vert 1_\xi \rangle = B^\ast(\xi) \vert 0 \rangle = \vert {1_\xi}_{t]} \rangle
\otimes \vert 0_{(t} \rangle + \vert 0_{t]} \rangle \otimes \vert
{1_\xi}_{(t} \rangle ,  \label{eq:factor-additive-1}
\end{equation}
where
\begin{equation}
\vert {1_\xi}_{t]} \rangle = B^{-\ast}_t(\xi) \vert 0_{t]} \rangle, \ \text{%
and} \ \vert {1_\xi}_{(t} \rangle = B^{+\ast}_t(\xi) \vert 0_{(t} \rangle ,
\end{equation}
and
\begin{equation}
B^-_t(\xi) = B(\xi \chi_{[0,t]}), \ \ B^+_t(\xi) = B(\xi \chi_{(t,\infty]}),
\ \ B(\xi) = B^-_t(\xi) + B^+_t(\xi) .  
\label{eq:B-xi-decomp-1}
\end{equation}
Here, $\chi_{[0,t]}$ is the indicator function for the time interval $[0,t]$%
. Note that while $\vert 1_\xi \rangle$ has unit norm, we have
\begin{equation}
\parallel \vert  {1_\xi}_{t]} \rangle \parallel^2 = \int_0^t \vert \xi(s) \vert^2
ds, \ \text{and} \ \parallel \vert {1_\xi}_{(t} \rangle \parallel^2 =
\int_t^\infty \vert \xi(s) \vert^2 ds .
\end{equation}

A consequence of the additive decomposition (\ref{eq:factor-additive-1}) and
the definitions (\ref{eq:B-xi-decomp-1}) is the following. Let $K(t)$ be a bounded operator acting on the full Hilbert space $\mathfrak{H}$ that is adapted, i.e. $K(t)$ acts trivially on $\mathfrak{F}_{(t}$, the field in the future.
Then the expectation with respect to the single photon field may be
expressed in terms of the vacuum state as follows:
\begin{eqnarray}
\mathbb{E}_{11} [ K(t) ] &=& \mathbb{E}_{00} [ B^-_t(\xi) K(t)
B^{-\ast}_t(\xi) + r(t) K(t) ]  
\label{eq:xi-phi}
\end{eqnarray}
where $r(t) = \int_t^\infty \vert \xi(s) \vert^2 ds$. 
 
\subsection{Master Equation}
\label{sec:master}

Before deriving the quantum filter, we work out dynamical equations for the
unconditioned single photon expectation,  \cite{GEPZ98}.
To assist us in evaluating this
expectation, we make use of the following lemma.

\begin{lemma}
\label{lemma:expectation-basic} 
Let $K(t)$ be a bounded quantum stochastic process
defined by
\begin{equation}
K(t) = \int_0^t M_0(s) ds + \int_0^t M_- (s) dB(s) + \int_0^t M_+(s)
dB^\ast(s) + \int_0^t M_1(s) d\Lambda(s),  
\label{eq:Kt-def}
\end{equation}
where $M_0$, $M_\pm$ and $M_1$ are bounded and adapted. Then we have
\begin{eqnarray}
\mathbb{E}_{11}[ K(t) ] &=& \mathbb{E}_{11}[ \int_0^t M_0(s) ds ] + \mathbb{E%
}_{10}[ \int_0^t M_- (s) \xi(s) ds ]  \notag \\
&& + \mathbb{E}_{01}[ \int_0^t M_+(s) \xi^\ast(s) ds ] + \mathbb{E}_{00}[
\int_0^t M_1 (s) \vert \xi(s) \vert^2 ds ] ,  \label{eq:Kt-11} \\
\mathbb{E}_{10}[ K(t) ] &=& \mathbb{E}_{10}[ \int_0^t M_0(s) ds ] + \mathbb{E%
}_{00}[ \int_0^t M_+ (s) \xi^\ast(s) ds ] ,  \label{eq:Kt-10} \\
\mathbb{E}_{01}[ K(t) ] &=& \mathbb{E}_{01}[ \int_0^t M_0(s) ds ] + \mathbb{E%
}_{00}[ \int_0^t M_- (s) \xi(s) ds ] ,  \label{eq:Kt-01} \\
\mathbb{E}_{00}[ K(t) ] &=& \mathbb{E}_{00}[ \int_0^t M_0(s) ds ] .
\label{eq:Kt-00}
\end{eqnarray}
\end{lemma}

\begin{proof}
Using (\ref{eq:xi-phi}), the expressions $B^-_t(\xi) = \int_0^t \xi^\ast(s)
dB(s)$, $B^{-\ast}_t(\xi) = \int_0^t \xi(s) dB^\ast(s)$, and the Ito rule we
have
\begin{eqnarray}
\mathbb{E}_{11}[ d K(t) ] &=& \mathbb{E}_{00} [ d ( B^-_t (\xi) K(t)
B^{-\ast}_t (\xi) + r(t) K(t) ) ]  \notag \\
&=& \mathbb{E}_{00} [ B^-_t (\xi) M_0(t) B^{-\ast}_t (\xi) + r(t) M_0(t)
\notag \\
&& + M_+(t) B^{-\ast}_t(\xi) \xi^\ast(t) + B^-_t (\xi) M_-(t) \xi(t) +M_1(t) |\xi(t)|^2] dt
\notag \\
&=& \mathbb{E}_{11} [ M_0(t) ] dt + \mathbb{E}_{00} [ M_+(t) B^{\ast} (\xi)
\xi^\ast(t) + B(\xi) M_-(t) \xi(t) + M_1 |\xi(t)|^2] dt .
\end{eqnarray}
This last line is justified since $M_\pm$ are adapted and $\mathbb{E}_{00}[
B^+_t(\xi) ] = 0$. That is,
\begin{equation}
\mathbb{E}_{11}[ dK(t) ] = \mathbb{E}_{11} [ M_0(t) ] dt + \mathbb{E}_{01}[
M_+(t) ] \xi^\ast(t) dt + \mathbb{E}_{10} [ M_-(t) ] \xi(t) dt +  \mathbb{E}_{00}[
\int_0^t M_1 (s) \vert \xi(s) \vert^2 ds ] .
\end{equation}
This proves (\ref{eq:Kt-11}). The remaining expressions are proven in a
similar manner.
\qed
\end{proof}

We will first express the master equation in Heisenberg form using the
expectations
\begin{equation}
\mu^{jk}_t(X) = \mathbb{E}_{jk}[ X(t) ].
\end{equation}
Note that for all $t \geq 0$ we have
\begin{equation}
\mu_t^{00}(I)= 1 = \mu^{11}_t(I) , \ \ \ \mu^{01}_t(I)=0=\mu^{10}_t(I).
\end{equation}

\begin{theorem}
\label{thm:master} The master equation in Heisenberg form for the system  
when the field is in the single photon state $\vert 1_\xi \rangle$ is given by
the system of equations
\begin{eqnarray}
\dot{\mu}^{11}_t (X) &=& \mu^{11}_t(\mathcal{L}(X)) + \mu^{01}_t( S^\ast
[X,L] ) \xi^\ast(t) + \mu^{10}_t( [L^\ast, X] S ) \xi(t)  \notag \\
&& + \mu^{00}_t( S^\ast X S - S) \vert \xi(t) \vert^2,
\label{eq:rho-dyn-a-11} \\
\dot{\mu}^{10}_t (X) &=& \mu^{10}_t(\mathcal{L}(X)) + \mu^{00}_t( S^\ast [X,
L] ) \xi^\ast(t) ,  \label{eq:rho-dyn-a-10} \\
\dot{\mu}^{01}_t (X) &=& \mu^{01}_t(\mathcal{L}(X)) + \mu^{00}_t( [L^\ast,
X] S ) \xi(t) ,  \label{eq:rho-dyn-a-01} \\
\dot{\mu}^{00}_t (X) &=& \mu^{00}_t(\mathcal{L}(X)) .
\label{eq:rho-dyn-a-00}
\end{eqnarray}
The initial conditions are
\begin{equation}
\mu^{11}_0(X)= \mu^{00}_0(X)= \langle \eta, X \eta \rangle, \ \
\mu^{10}_0(X)= \mu^{01}_0(X)=0.
\end{equation}
\end{theorem}

\begin{proof}
Equations (\ref{eq:rho-dyn-a-11})- (\ref{eq:rho-dyn-a-00}) are obtained by
applying Lemma \ref{lemma:expectation-basic} to the Heisenberg equation (\ref
{eq:X-dyn}).
\qed
\end{proof}

It is apparent from Theorem \ref{thm:master} that the single photon
expectation $\mu^{11}_t(X) = \mathbb{E}_{11}[ X(t)]$ cannot be determined by
a single differential equation, and that instead a system of coupled equations is
required, equation (\ref{eq:rho-dyn-a-11})-(\ref{eq:rho-dyn-a-00}).  Note
that the unitary matrix $S$ appearing in the Schr\"{o}dinger equation (\ref{eq:unitary}) does appear in the single photon master equations (\ref{eq:rho-dyn-a-11})-(\ref{eq:rho-dyn-a-00}), in contrast to the vacuum case (which corresponds to (\ref{eq:rho-dyn-a-00})).

\subsection{Embedding}
\label{sec:the-matrix}

 In this section we construct a suitable embedding for the system
and single photon field, and show how the system of master equations from
Section \ref{sec:master} can be compactly represented as a single 
equation for a larger system. This embedding will be used in subsequent
sections to derive the quantum filter. We should emphasize, however, that our embedding is
not the same as that used in \cite{HPB04}, \cite{GJN11a}, \cite{GJNC11}. The embedding is illustrated in Figure \ref
{fig:extend-1}.

\begin{figure}[h]
\begin{center}
\includegraphics{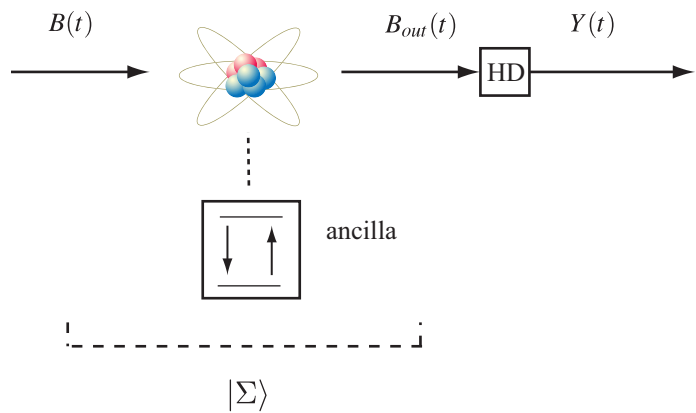} 
\end{center}
\caption{System  embedded in the extended system. While the analysis does not employ  any coupling between the system and ancilla two-level system, the ancilla, system and field are assumed to be initialized in a superposition state $\vert \Sigma \rangle$ defined in equation (\ref{eq:super}).}
\label{fig:extend-1}
\end{figure}

Recall that the system and field are defined on a Hilbert space $\mathfrak{H}
= \mathfrak{H}_S \otimes \mathfrak{F}$. We define an extended space
\begin{equation}
\tilde {\mathfrak{H}}= \mathbb{C}^2 \otimes \mathfrak{H} = \mathfrak{H}
\oplus \mathfrak{H} 
\end{equation}
which includes the system, field and an ancilla two-level system.
Let $\vert e_0 \rangle$ and $\vert e_1 \rangle$ be an orthonormal basis for $\mathbb{C}^2$,
\begin{equation}
\vert e_0  \rangle = \left[
\begin{array}{c}
0 \\
1
\end{array}
\right], \ \ \vert e_1 \rangle  = \left[
\begin{array}{c}
1 \\
0
\end{array}
\right],
\end{equation}
and let $A$ be an operator acting on $\mathbb{C}^2$, i.e. a complex $2
\times 2$ matrix
\begin{equation}
A = \left[
\begin{array}{cc}
a_{11} & a_{10} \\
a_{01} & a_{00}
\end{array}
\right].
\end{equation}
It may be helpful to think of operators $A \otimes X \otimes F$ on the
extended space $\tilde {\mathfrak{H}}$ represented in the Kronecker product
form
\begin{equation}
A \otimes (X \otimes F) = \left[
\begin{array}{cc}
a_{11} (X \otimes F) & a_{10} (X \otimes F) \\
a_{01} (X \otimes F) & a_{00} (X \otimes F)
\end{array}
\right].
\end{equation}

We allow the extended system to evolve unitarily according to $I \otimes U(t)
$, where $U(t)$ is the unitary operator for the system and field, given by
the Schr\"{o}dinger equation (\ref{eq:unitary}). Note in particular that the
system is not coupled to the ancilla $\mathbb{C}^2$, and observables of this
two-level system are static. We initialize the extended system in
the superposition state
\begin{equation}
\vert \Sigma \rangle = \alpha_1 \vert e_1 \eta 1_\xi \rangle + \alpha_0 \vert
e_0 \eta 0 \rangle,  \label{eq:super}
\end{equation}
where $\vert \alpha_0 \vert^2 + \vert \alpha_1 \vert^2 =1$. This state
evolves according to
\begin{equation}
\vert \Sigma(t) \rangle = (I \otimes U(t)) \vert \Sigma \rangle.
\end{equation}
For notational convenience we write
\begin{equation}
w_{jk} = \alpha_j^\ast \alpha_k  \label{w_jk}
\end{equation}
and note that $w = \sum_{jk} w_{jk}  \vert e_j \rangle \langle e_k \vert$ is a density matrix for $\mathbb{C}^2
$.

The expectation with respect to the superposition state $\vert \Sigma \rangle$ is given by
\begin{equation}
\tilde \mu_t( A \otimes X) = \mathbb{E}_{\psi} [ A \otimes X(t) ] = \langle \Sigma \vert
(A \otimes X(t)) \vert \Sigma \rangle = \sum_{jk} w_{jk} a_{jk} \mu^{jk}_t(X).
\label{eq:master-mu-rep1}
\end{equation}
This expectation is correctly normalized, $\mu_t(I \otimes I)=1$, and the
expectations $\mu^{jk}_t(X)$ defined in Section \ref{sec:master} are scaled
components of $\tilde \mu_t(A\otimes X)$:
\begin{equation}
\mu^{jk}_t(X) = \frac{ \tilde \mu_t( \vert e_j \rangle \langle e_k \vert  \otimes X)}{ w_{jk} },
\label{eq:master-mu-rep2}
\end{equation}
for $w_{jk} \neq 0$, otherwise it can be set to, say, 0. We also have
\begin{equation}
\mu^{jk}_t(X) = \frac{ w_{11}  \tilde\mu_t( \vert e_j \rangle \langle e_k \vert  \otimes X) }{ w_{jk}  \tilde \mu_t(  \vert e_1 \rangle \langle e_1 \vert \otimes I) } .
\label{eqmui-jk-bayes}
\end{equation}

Note that in the extended space the Schr\"{o}dinger and Heisenberg pictures
are related by
\begin{equation}
\mathbb{E}_{\Sigma(t) } [ A \otimes X \otimes F ] = \mathbb{E}_{\Sigma} [ A
\otimes U^\ast(t) (X \otimes F) U(t) ] .
\end{equation}

In order to derive the equation for expectations in the extended system, we
need the following lemma, which follows from Lemma \ref
{lemma:expectation-basic}, and makes use of the matrices
\begin{equation}
\sigma_+ =  \vert e_1 \rangle \langle e_0 \vert = \left[
\begin{array}{cc}
0 & 1 \\
0 & 0
\end{array}
\right], \ \ \sigma_- =  \vert e_0 \rangle \langle e_1 \vert  = \left[
\begin{array}{cc}
0 & 0 \\
1 & 0
\end{array}
\right].
\end{equation}

\begin{lemma}
\label{lemma:expectation-basic-extended} 
Assume $\alpha_0\neq 0$, and let $M(t)$ be bounded and adapted. Then
\begin{eqnarray}
\mathbb{E}_\Sigma [ \int_0^t A \otimes M(s) dB(s) ] &=& \nu \mathbb{E}_\Sigma [
\int_0^t ( A \sigma_+)\otimes M(s) \xi(s) ds ], \\
\mathbb{E}_\Sigma [ \int_0^t A \otimes M (s) dB^\ast(s) ] &=& \nu^\ast \mathbb{%
E}_\Sigma [ \int_0^t (\sigma_-A )\otimes M(s) \xi^\ast(s) ds], \\
\mathbb{E}_\Sigma [ \int_0^t A \otimes M(s) d\Lambda(s) ] &=& \vert \nu
\vert^2 \mathbb{E}_\Sigma [ \int_0^t ( \sigma_- A \sigma_+ )\otimes M(s) \vert
\xi(s) \vert^2 ds ],
\end{eqnarray}
where
\begin{equation}
\nu = \frac{\alpha_1}{\alpha_0}.
\end{equation}
\end{lemma}

In Lemma \ref{lemma:expectation-basic-extended}, expectations of stochastic integrals with respect to the superposition state $\vert \Sigma \rangle$ are expressed  in terms of expectations of non-stochastic integrals again with respect to $\vert \Sigma \rangle$ with the aid of the matrices $\sigma_\pm$ acting on the ancilla system $\mathbb{C}^2$.  The action of the field annihilation, creation and gauge operators is therefore captured algebraically and all expectations in these relations are with respect to the same state.

We now have

\begin{theorem}
\label{thm:master-matrix} Assume $\alpha_0\neq 0$. Then the expectation $\tilde \mu_t(A \otimes X)$ (defined by (\ref{eq:master-mu-rep1}))
evolves according to
\begin{equation}
\dot {\tilde \mu}_t(A \otimes X) = \tilde \mu_t (\mathcal{G}_t(A \otimes X) ),
\label{eq:matrix-master}
\end{equation}
where
\begin{eqnarray}
\mathcal{G}_t(A \otimes X) &=& A \otimes \mathcal{L}(X) + ( A\sigma_+)
\otimes [L^\ast, X] S \nu \xi(t) + (\sigma_- A ) \otimes S^\ast [X,L] \nu^\ast
\xi^\ast(t)  \notag \\
&& + ( \sigma_- A \sigma_+ ) \otimes (S^\ast X S - X)) \vert \nu \xi(t)
\vert^2.
\end{eqnarray}
\end{theorem}

The reader may easily verify that the system of  master
equations (\ref{eq:rho-dyn-a-11})-(\ref{eq:rho-dyn-a-00}) for $\mu^{jk}_t(X)$%
, $j,k=1,0$, follows from equation (\ref{eq:matrix-master}%
) by setting $A=\vert e_j \rangle \langle e_k \vert$.

\subsection{Quantum Filter for the Extended System}
\label{sec:matrix-filter}

The extended system provides a convenient  framework for quantum
filtering, since all expectations can be expressed in terms of  the superposition state $\vert \Sigma \rangle$.
 Our immediate goal in this section is to determine 
the equation for the quantum conditional expectation
\begin{equation}
\tilde \pi_t( A \otimes X) = \mathbb{E}_{\Sigma}[ A \otimes X(t) \, \vert \, I
\otimes \mathscr{Y}_t] ,
\label{eq:matrix-c-exp-def}
\end{equation}
and in  Section \ref{sec:single-photon-filter} we will explain how the quantum filter for
the single photon field may be obtained from this equation.

The continuously monitored field observable that corresponds to the
conditional expectation (\ref{eq:matrix-c-exp-def}) is $I \otimes Y(t)$, and
from (\ref{eq:Y-out}) we have the corresponding output equation for the
extended system:
\begin{equation}
d(I \otimes Y(t)) = I \otimes (L(t)+L^\ast(t)) dt + I \otimes (S(t) dB(t) +
S^\ast(t) dB^\ast(t) ) .
\end{equation}
In what follows we will make use of the following lemma concerning
expectations of the process
\begin{equation}
V(t) = \int_0^t ( S(s) dB(s) + S^\ast(s) dB^\ast(s) )
\end{equation}
 with respect to the single photon state.

\begin{lemma}
\label{lemma:Z-compensated-mtg} For any $K \in \mathscr{Y}_s$, we have
\begin{eqnarray}
\mathbb{E}_{11}[ (V(t) - V(s)) K ] &=& \mathbb{E}_{10}[ \int_s^t S(r)
\xi(r)dr K] + \mathbb{E}_{01}[ \int_s^t S^\ast(r) \xi^\ast(r)dr K] .
\label{eq:V-compensated} 
\end{eqnarray}
\end{lemma}

\begin{proof}
Equation (\ref{eq:V-compensated}) is obtained using the additive
decomposition (\ref{eq:factor-additive-1}) and Lemma \ref
{lemma:expectation-basic}.
\qed
\end{proof}

Note that an operator $K$ in the unital commutative algebra $I \otimes %
\mathscr{Y}_t$ has the form $K=I \otimes \tilde K$, where $\tilde K \in %
\mathscr{Y}_t$. By the spectral theorem, \cite[Theorem 3.3]{BHJ07}, we may identify $K$ and $\tilde K$, both of
which are equivalent to a classical stochastic process $K_t(s)$, $0 \leq s
\leq t$. In the remainder of this paper, we use these identifications without further comment. 
The quantum conditional expectation $\tilde \pi_t( A \otimes X) \in I
\otimes \mathscr{Y}_t$ is well defined because $A \otimes X(t)$ is in the
commutant $I \otimes \mathscr{Y}_t^{\prime}$ of the algebra $I \otimes %
\mathscr{Y}_t$, and is characterized by the requirement that
\begin{equation}
\mathbb{E}_{\Sigma} [ \tilde \pi_t(A \otimes X) I \otimes K ] = \mathbb{E}_{\Sigma} [
(A \otimes X(t))( I \otimes K) ]  \label{eq:c-exp-def-matrix}
\end{equation}
for all $K \in \mathscr{Y}_t$, see, e.g. \cite[Definition 3.13]{BHJ07}.

\begin{theorem}
\label{thm:matrix-filter} 
Assume $\alpha_0 \neq 0$. The conditional expectation $\tilde \pi_t(A \otimes X)$ defined by
(\ref{eq:matrix-c-exp-def}) for the extended system satisfies
\begin{eqnarray}
d\tilde \pi_t(A\otimes X) &=&\tilde  \pi_t( \mathcal{G}_t(A \otimes X) ) dt + \mathcal{H}%
_t(A\otimes X) dW(t),  
\label{pi_filter}
\end{eqnarray}
where
\begin{eqnarray}
\mathcal{H}_t(A \otimes X) &=& \tilde  \pi_t( A \otimes (XL+L^\ast X) ) - \tilde \pi_t( A
\otimes X) \pi_t( I \otimes (L+L^\ast) )  \notag \\
&& + \tilde \pi_t( ( A \sigma_+) \otimes XS) \nu \xi(t) + \tilde \pi_t( ( \sigma_-A )
\otimes S^\ast X) \nu^\ast \xi^\ast(t)  \notag \\
&& - \tilde \pi_t( A \otimes X) \tilde \pi_t( ( \sigma_+ \otimes S) \nu \xi(t) +( \sigma_-
\otimes S^\ast ) \nu^\ast \xi^\ast(t) )  \label{eq:matrix-filter}
\end{eqnarray}
and
\begin{equation}
dW(t) = dY(t) - \tilde \pi_t( I \otimes (L+L^\ast) + ( \sigma_+ \otimes S) \nu
\xi(t) + (\sigma_- \otimes S^\ast) \nu^\ast \xi^\ast(t) ) dt .
\label{eq:innovation-matrix}
\end{equation}
The process $W(t)$ defined by (\ref{eq:innovation-matrix}) is a $I \otimes %
\mathscr{Y}_t$ Wiener process  with respect to $\vert \Sigma \rangle$ and is called the \emph{%
innovations process}.
\end{theorem}

\begin{proof}
We follow the characteristic function method \cite{HSM05}, \cite{VPB93}, \cite{VPB92},
 whereby we postulate that the filter has
the form
\begin{equation}
d \tilde \pi_t( A \otimes X) = \mathcal{F}_t(A \otimes X) dt + \mathcal{H}_t(A
\otimes X) I\otimes dY(t) ,
\end{equation}
where $\mathcal{F}_t$ and $\mathcal{H}_t$ are to be determined.

Let $f$ be square integrable, and define a process $c_f$ by
\begin{equation}
dc_f(t) = f(t) c_f(t) dY(t), \ \ c_f(0)=1.
\end{equation}
Then $I \otimes c_f(t)$ is adapted to $I \otimes \mathscr{Y}_t$, and the
defining relation (\ref{eq:c-exp-def-matrix}) implies that
\begin{equation}
\mathbb{E}_\Sigma [ A \otimes (X(t) c_f(t) ) ] = \mathbb{E}_\Sigma [ \tilde \pi_t( A
\otimes X) I \otimes c_f(t) ) ]
\end{equation}
holds for all $f$. By calculating the differentials of both sides, taking
expectations and conditioning we obtain
\begin{eqnarray}
\mathbb{E}_\Sigma [ A \otimes (dX(t) c_f(t)) ] &=& \mathbb{E}_\Sigma [ (I
\otimes c_f(t) ) \tilde \pi_t( \mathcal{G}(A \otimes X) ) \\
&& + (I \otimes f(t) c_f(t) ) \{\tilde  \pi_t( A \otimes (XL+L^\ast X) )  \notag \\
&& + \tilde \pi_t( A \sigma_+ \otimes XS ) \nu \xi(t) + \tilde \pi_t( \sigma_-A \otimes
S^\ast X) \nu^\ast \xi^\ast(t) \} ]dt  \notag
\end{eqnarray}
and
\begin{eqnarray}
&& \mathbb{E}_\Sigma [ A \otimes (d\tilde \pi_t(A\otimes X) c_f(t)) ] 
\\
 &=& \mathbb{E}_\Sigma [ (I \otimes c_f(t) \{ \mathcal{F}_t(A\otimes X) + \mathcal{H}_t(A\otimes X) \tilde \pi_t( I \otimes
(L+L^\ast))\nonumber  \\
&& +\mathcal{H}_t(A\otimes X) \pi_t( ( \sigma_+ \otimes S) \nu \xi(t) + (\sigma_-
\otimes S^\ast ) \nu^\ast \xi^\ast(t) ) \}  \notag \\
&& + (I \otimes f(t) c_f(t) ) \{ \tilde \pi_t( A \otimes X) \tilde \pi_t( I \otimes
(L+L^\ast) ) + \mathcal{H}_t  \notag \\
&& + \tilde \pi_t( A \otimes X)\tilde  \pi_t( \sigma_+ \otimes S \nu \xi(t) + \sigma_-
\otimes S^\ast \nu^\ast \xi^\ast(t) ) \} ] dt.  \notag
\end{eqnarray}
Now equating coefficients of $c_f(t)$ and $f(t) c_f(t)$ we solve for $%
\mathcal{F}_t(A\otimes X)$ and $\mathcal{H}_t(A\otimes X)$ to obtain the filter equation.

We now prove the martingale property $\mathbb{E}_{\Sigma} [ I \otimes ( W(t) -
W(s) ) \, \vert \, I \otimes \mathscr{Y}_s ] = 0$, that is,
$
\mathbb{E}_{\Sigma} [ I \otimes ( W(t) - W(s) ) (I \otimes K) ] =0 
$
for all $K \in \mathscr{Y}_t$. Now
\begin{eqnarray}
&& \mathbb{E}_{\Sigma} [ I \otimes ( W(t) - W(s) ) (I \otimes K) ]  \notag \\
&=& \mathbb{E}_{\Sigma} [ \{ I \otimes (Y(t) - Y(s) )  \notag \\
&& - \int_s^t \pi_r( I \otimes (L+L^\ast) + \sigma_+ \otimes S \xi(r) +
\sigma_- \otimes S^\ast \xi^\ast(r) ) dr \} I \otimes K ]  \notag \\
&=& \mathbb{E}_{\Sigma} [ \{ I \otimes ( Y(t) - Y(s) )  \notag \\
&& - \int_s^t ( I \otimes (L(r)+L^\ast(r) ) + \sigma_+ \otimes S(r) \xi(r) +
\sigma_- \otimes S^\ast(r) \xi^\ast(r) ) dr \} I \otimes K ]  \notag \\
&=& \mathbb{E}_{\Sigma} [ \{ I \otimes (V(t) - V (s)) - \int_s^t ( \sigma_+
\otimes S(r) \nu \xi(r) + \sigma_- \otimes S^\ast(r) \nu^\ast \xi^\ast(r) )
dr \} I \otimes K ] =0.  \notag
\end{eqnarray}
To see that this last expression is zero, we make use of Lemma \ref
{lemma:Z-compensated-mtg},  the multiplicative factorization of the vacuum
state, the fact that $V(t)$ has zero expectation in the vacuum state, to
find that
\begin{eqnarray}
\mathbb{E}_{\Sigma} [ I \otimes (V(t) - V(s)) I \otimes  K ] &=& w_{11} \mathbb{E}_{11}[
(V(t) - V(s)) K ] + w_{00} \mathbb{E}_{00}[ (V(t) - V(s)) K ]  \notag \\
&=& w_{11} ( \int_s^t \mathbb{E}_{10}[ S(r)K] \xi(r)dr + \int_s^t \mathbb{E}%
_{01}[ S^\ast(r) K] \xi^\ast(r)dr ) ,  \notag
\end{eqnarray}
and
\begin{eqnarray}
&& \mathbb{E}_{\Sigma} [ \int_s^t ( \sigma_+ \otimes S(r)K) \xi(t) + \sigma_-
\otimes S^\ast (r) K \xi^\ast(t) ) dr ]  \notag \\
&=&w_{11} ( \int_s^t \mathbb{E}_{10}[ S(r) K] \xi(r) dr + \int_s^t \mathbb{E}%
_{01}[ S^\ast (r) K] \xi^\ast(r) dr ) .  \notag
\end{eqnarray}
Finally, since $dW(t) dW(t)=dt$, Levy's Theorem implies that $W(t)$ is a $\mathscr{Y}_t$ Wiener process.
This completes the proof.
\qed
\end{proof}

Notice the terms involving $\sigma_\pm$ in the filter (equation (\ref
{eq:matrix-filter})) and in the innovations process (equation (\ref
{eq:innovation-matrix})). These terms arise from expectations involving the
single photon state. Note that due to the martingale property of the
innovations process $W(t)$ we see that if we take the expected value of
equation (\ref{eq:matrix-filter}) we recover equation (\ref
{eq:matrix-master}), consistent with $\mathbb{E}_\Sigma [ \tilde \pi_t(A\otimes X) ] = \tilde \mu_t(A\otimes X)
$ and the definition of conditional expectation.

\subsection{Single Photon Quantum Filter}
\label{sec:single-photon-filter}

We return now to the main goal of the paper, namely the determination of the
quantum filter for the conditional state when the field is in the single
photon state, as stated in equation (\ref{eq:cond-exp}). As discussed
earlier, our strategy is to make use of the filtering results obtained in
Section \ref{sec:matrix-filter}   for the extended
system.

\begin{lemma}
Assume $\alpha_0 \neq 0$. Define the conditional quantities $\pi _{t}^{jk}(X)$ by
\begin{equation}
\pi^{jk}_t(X) = \frac{ w_{11} \tilde  \pi_t(  \vert e_j \rangle \langle e_k \vert \otimes X) }{ w_{jk}  \tilde \pi_t(  \vert e_1 \rangle \langle e_1 \vert \otimes I) } .
\label{eq:pi-jk-def}
\end{equation}
where $\tilde \pi_t(A \otimes X)$ is the conditional state for the extended system defined by (\ref{eq:matrix-c-exp-def}).
Then for
all $K \in \mathscr{Y}_{t}$ we have
\begin{equation}
\mathbb{E}_{11}[\pi _{t}^{jk}(X)\,K]=\mathbb{E}%
_{jk}[j_{t}(X)K].  \label{eq:pi-jk-c-exp-0}
\end{equation}
\end{lemma}

\begin{proof}
We have
\begin{eqnarray*}
\mathbb{E}_{11}[ \pi^{jk}_t(X) K ] &=&  \frac{1}{w_{11}} \mathbb{E}_{\Sigma} [  \vert e_1 \rangle \langle e_1 \vert \otimes (\pi^{jk}_t(X)K) ]
\\
&=&
\frac{1}{w_{11}} \mathbb{E}_{\Sigma} [  \pi_t( \vert e_1 \rangle \langle e_1 \vert \otimes I)  (I \otimes  \pi^{jk}_t(X)K) ]
\\
&=&
\frac{1}{w_{jk}} \mathbb{E}_{\Sigma} [  \pi_t( \vert e_j \rangle \langle e_k \vert  \otimes X)  (I \otimes  K) ] 
\\
&=&
\frac{1}{w_{jk}} \mathbb{E}_{\Sigma} [ (  \vert e_j \rangle \langle e_k \vert \otimes j_t(X))  (I \otimes  K) ] 
\\
&=&
\mathbb{E}_{jk}[ j_t(X) K ]
\end{eqnarray*}
as required.
\qed
\end{proof}

 We  can now present our main theorem for the quantum filter
for the single photon field state.

\begin{theorem}
\label{thm:main} The quantum filter for the conditional expectation with
respect to the single photon field is given in the Heisenberg picture by
\begin{equation}
\hat X(t)= \mathbb{E}_{11}[ X(t) \, \vert \, \mathscr{Y}_t ] = \pi^{11}_t(X),
\end{equation}
where $\pi^{11}_t(X)$ is defined by (\ref{eq:pi-jk-def}) (for $j=k=1$),  and is given by the
system of equations
\begin{eqnarray}
d\pi^{11}_t (X) &=& (\pi^{11}_t(\mathcal{L}(X)) + \pi^{01}_t( S^\ast [X,L] )
\xi^\ast(t) + \pi^{10}_t( [L^\ast, X] S ) \xi(t)  \notag \\
&& + \pi^{00}_t( S^\ast X S - X) \vert \xi(t) \vert^2)dt  \notag \\
&& + ( \pi^{11}_t( XL + L^\ast X) + \pi^{01}_t(S^\ast X) \xi^\ast(t) +
\pi^{10}_t( XS) \xi(t)  \notag \\
&& - \pi^{11}_t(X) ( \pi^{11}_t(L+L^\ast) + \pi^{01}_t(S) \xi(t) +
\pi^{10}_t(S^\ast) \xi^\ast(t) ) )dW(t),  
\label{eq:pi-dyn-a-11} 
\\
d\pi^{10}_t (X) &=& ( \pi^{10}_t(\mathcal{L}(X)) + \pi^{00}_t( S^\ast [X, L]
) \xi^\ast(t) )dt  \notag \\
&& + ( \pi^{10}_t( XL + L^\ast X) + \pi^{00}_t(S^\ast X) \xi^\ast(t)  \notag
\\
&& - \pi^{10}_t(X) ( \pi^{11}_t(L+L^\ast) + \pi^{01}_t(S) \xi(t) +
\pi^{10}_t(S^\ast) \xi^\ast(t) ) )dW(t),  
\label{eq:pi-dyn-a-10} 
\\
d\pi^{01}_t (X) &=& ( \pi^{01}_t(\mathcal{L}(X)) + \pi^{00}_t( [L^\ast, X] S
) \xi(t) )dt  \notag \\
&& + ( \pi^{01}_t( XL + L^\ast X) + \pi^{00}_t( XS) \xi(t)  \notag \\
&& - \pi^{01}_t(X) ( \pi^{11}_t(L+L^\ast) + \pi^{01}_t(S) \xi(t) +
\pi^{10}_t(S^\ast) \xi^\ast(t) ) )dW(t),  
\label{eq:pi-dyn-a-01} 
\\
d\pi^{00}_t (X) &=& \pi^{00}_t(\mathcal{L}(X)) dt + ( \pi^{00}_t( XL +
L^\ast X)  \notag \\
&& - \pi^{00}_t(X) ( \pi^{11}_t(L+L^\ast) + \pi^{01}_t(S) \xi(t) +
\pi^{10}_t(S^\ast) \xi^\ast(t) ) )dW(t) . 
 \label{eq:pi-dyn-a-00}
\end{eqnarray}
Here, the innovations process $W(t)$ is a $\mathscr{Y}_t$ Wiener process with
respect to the single photon state and is defined by
\begin{equation}
dW(t) = dY(t) - ( \pi^{11}_t( L+L^\ast) + \pi^{10}_t(S) \xi(t) +
\pi^{01}_t(S^\ast) \xi^\ast(t) ) dt .  \label{eq:innovation-11}
\end{equation}
The initial conditions are
\begin{equation}
\pi^{11}_0(X)= \pi^{00}_0(X)= \langle \eta, X \eta \rangle, \ \
\pi^{10}_0(X)= \pi^{01}_0(X)=0.
\label{eq:pi-jk--initial}
\end{equation}
\end{theorem}

\begin{proof}
Suppose first that $\alpha_0 \neq 0$. Setting $j=k=1$ in equation (\ref{eq:pi-jk-c-exp-0}) above, and noting that $%
K\in \mathscr{Y}_{s}$ was otherwise arbitrary, we deduce that $\pi
_{t}^{11}(X)$ is the desired conditional expectation for the single photon
field state, as characterized by equation (\ref{eq:c-exp-def}). 
The differential equations (\ref{eq:pi-dyn-a-11})-(\ref{eq:pi-dyn-a-00}) follow  from the definition (\ref{eq:pi-jk-def}), the filter (\ref{pi_filter}) for the extended system, and the Ito rule. Next, we note that the coefficients of the QSDEs (\ref{eq:pi-dyn-a-11})-(\ref{eq:pi-dyn-a-00}), the initial conditions, and $Y_t$ do not depend on $\alpha_0$ and $\alpha_1$. Hence, the solutions $\pi_t^{jk}(X)$ of this system of  equations are independent of $\alpha_0$ and $\alpha_1$. Therefore, $\pi_t^{jk}(X)$ can be defined for $\alpha_j \in \{0,1\}$, $j=0,1$, and is in fact identical for all $0 \leq |\alpha_0|,|\alpha_1| \leq 1$.

We now prove that $W(t)$ is a $\mathscr{Y}_t$-martingale, that is, $\mathbb{E}_{11}[W(t)-W(s)\,|\mathscr{Y}_{s}]=0$.
To this end, let $K \in \mathscr{Y}_{s}$. Then
\begin{eqnarray*}
\mathbb{E}_{11}[\{W(t)-W(s)\}K] &=&\mathbb{E}_{11}[\{Y(t)-Y(s)  \notag \\
&&-\int_{s}^{t}(\pi _{r}^{11}((L+L^{\ast })+\pi _{r}^{10}(S)\xi (t)+\pi
_{r}^{01}(S^{\ast })\xi ^{\ast }(t))dr\}K]  \notag \\
&=&\mathbb{E}_{11}[\{\int_{s}^{t}(L(r)+L^{\ast }(r))dr+V(t)-V(s)  \notag \\
&&-\int_{s}^{t}(\pi _{r}^{11}((L+L^{\ast })+\pi _{r}^{10}(1)\xi (t)+\pi
_{r}^{01}(1)\xi ^{\ast }(t))dr\}K]  \notag \\
&=&\mathbb{E}_{11}[\{\int_{s}^{t}(L(r)+L^{\ast }(r)-\pi _{r}^{11}((L+L^{\ast
}))dr\}K]  \notag \\
&&+\mathbb{E}_{11}[\{V(t)-V(s)-\int_{s}^{t}(\pi _{r}^{10}(S)\xi (t)+\pi
_{r}^{01}(S^{\ast })\xi ^{\ast }(t))dr\}K]
\end{eqnarray*}
however, this vanishes from (\ref{eq:pi-jk-c-exp-0}), Lemma \ref{lemma:Z-compensated-mtg},
and
\begin{eqnarray*}
&& \mathbb{E}_{11}[\{\int_{s}^{t}(\pi _{r}^{10}(S)\xi (t)+\pi
_{r}^{01}(S^{\ast })\xi ^{\ast }(t))dr\}K]
\\
&& =\mathbb{E}_{10}[\int_{s}^{t}S(r)\xi (r)dr K]+\mathbb{E}_{01}[%
\int_{s}^{t}S^{\ast }(r)\xi ^{\ast }(r)dr K ].
\end{eqnarray*}

Finally, since $dW(t) dW(t)=dt$, Levy's Theorem implies that $W(t)$ is a $\mathscr{Y}_t$ Wiener process.
\qed
 \end{proof}

\subsection{Reference Method for Filtering in the Extended System}
\label{sec:extended-filter}

The reference method is one of the standard approaches to filtering theory,
with its origins in the work of Duncan, Mortensen, Zakai, Holevo and
Belavkin, see \cite{RE82}, \cite{AH91}, \cite{VPB92}, \cite{BH06}, \cite{BHJ07}. In this section we
apply this approach, as described in \cite[sec. 6]{BHJ07}, to the filtering
problem in the extended system, giving an independent derivation of the fundamental filtering equation.

Our first step is the following.

\begin{lemma}
\label{lemma:matrix-F-rep} Assume $\alpha_0 \neq 0$. Then we have
\begin{equation}
\mathbb{E}_{\Sigma}[ (A \otimes X(t))  ] = \mathbb{E}_{\Sigma}[
F^\ast(t) (A \otimes X) F(t) ]
\end{equation}
  where $F(t) \in (I\otimes %
\mathscr{Z}_t)^{\prime}$ is given by
\begin{eqnarray}
dF(t) &=& ( G_0(t) dt +G_1(t)dZ(t)) F(t) ,  \label{dF}
\end{eqnarray}
$F(0)=I$, and 
\begin{eqnarray}
G_0(t) &=& -I \otimes ( \frac{1}{2} L^\ast L + iH) - \sigma_+ \otimes
(L+L^\ast S) \nu \xi(t) , \\
G_1(t) &=& I \otimes L + \sigma_+ \otimes (S-I) \nu \xi(t) .
\end{eqnarray}
\end{lemma}

\begin{proof}
Let us suppose that $F\left( t\right) $ satisfies (\ref{dF}) with
coefficients $G_{0}\left( t\right) $, $G_{1}\left( t\right) $ which both
commute with $\sigma _{+}\otimes I$, that is,
\begin{equation*}
G_{i}\left( t\right) =I\otimes g_{i0}\left( t\right) +\sigma _{+}\otimes
g_{i1}\left( t\right) .
\end{equation*}
Then
\begin{equation*}
\mathbb{E}_{\Sigma }\left[ d\left\{ F\left( t\right) ^{\ast }A \otimes XF\left( t\right)
\right\} \right] = \mathbb{E}_{\Sigma }\left[ F\left( t\right) ^{\ast }\left\{
T_t (A \otimes X) dt +Q_t (A \otimes X)dZ(t) \right\} F\left( t\right) \right] ,
\end{equation*}
where
\begin{eqnarray}
T_{t}\left( A \otimes X\right) &=&G_1 (t)^{\ast }A \otimes XG_1 (t)+A \otimes XG_0 (t)+G_0 (t)^{\ast }A \otimes X, \\
Q_{t}\left( A \otimes X\right) &=&A \otimes XG_1 (t)+G_1 (t)^{\ast }A \otimes X.
\end{eqnarray}
Using lemma \ref{lemma:expectation-basic-extended} we see that this equals
\begin{equation*}
\mathbb{E}_{\Sigma}[F\left( t\right) ^{\ast }\{ T_t (A \otimes X) + Q_t (A \otimes X) \nu \xi
\left( t\right) \left( \sigma _{+}\otimes I\right) +\nu ^{\ast }\xi \left(
t\right) ^{\ast }\left( \sigma _{-}\otimes I\right) Q_t (A \otimes X) \}F\left(
t\right) ]dt.
\end{equation*}
We now require that this equals $\mathbb{E}_{\Sigma }\left[ \mathcal{G}%
_t\left( A \otimes X\right) \right] dt$. This implies the four identities
\begin{eqnarray}
g_{11}{}^{\ast }Xg_{11}+\nu \xi g_{11}^{\ast }X+\nu ^{\ast }\xi ^{\ast
}Xg_{11} &=&|\nu \xi |^{2}(S^{\ast }XS-X)  \label{id1} \\
g_{10}^{\ast }Xg_{11}+Xg_{01}+\nu \xi g_{10}^{\ast }X+\nu \xi Xg_{10} &=&\nu
\xi \lbrack L^{\ast },X]S  \label{id2} \\
g_{11}^{\ast }Xg_{10}+g_{01}{}^{\ast }X+\nu ^{\ast }\xi ^{\ast
}g_{10}{}^{\ast }X+\nu ^{\ast }\xi ^{\ast }Xg_{10} &=&\nu ^{\ast }\xi ^{\ast
}S^{\ast }[X,L]  \label{id3} \\
g_{10}^{\ast }Xg_{10}+g_{00}^{\ast }X+Xg_{00} &=&\mathcal{L}\left( X\right) .
\label{id4}
\end{eqnarray}
The first identity (\ref{id1}) is satisfied if $g_{11}=\nu \xi (S-1)$.
Substituting $X=I$ into (\ref{id2}) we deduce that $g_{01}=-\nu \xi
g_{10}^{\ast }S-\nu \xi g_{10}$ and thus
\begin{equation*}
g_{10}^{\ast }Xg_{11}+Xg_{01}+\nu \xi g_{10}^{\ast }X+\nu \xi Xg_{10}=\nu
\xi \lbrack g_{10}^{\ast },X]S.
\end{equation*}
Therefore (\ref{id2}) is satisfied if $g_{10}=L$, and consequently $%
g_{01}=-\nu \xi \left( L+L^{\ast }S\right) $. It then follows that (\ref{id3}%
) will be automatically satisfied, while (\ref{id4}) then only requires that
$g_{00}=-(\frac{1}{2}L^{\ast }L+iH)$ in order to obtain the Lindblad
generator $\mathcal{L}\left( X\right) $.

This leads us precisely to the coefficients $G_i (t)$ stated in the lemma,
and the identity
\begin{equation}
T_t (A \otimes X) + Q_t (A \otimes X) \nu \xi \left( t\right) \left( \sigma _{+}\otimes I\right)
+\nu ^{\ast }\xi \left( t\right) ^{\ast }\left( \sigma _{-}\otimes I\right)
Q_t (A \otimes X) = \mathcal{G}_t\left( A \otimes X\right)
\end{equation}

\end{proof}

The following Bayes' relation is proven along similar lines to \cite[Theorem
6.2]{BHJ07}.

\begin{lemma}
\label{lemma:matrix-bayes} For $\alpha_0 \neq 0$, define
\begin{eqnarray}
\varsigma_t( A \otimes X) = (I \otimes U^\ast(t)) \mathbb{E}_{\Sigma} [
F^\ast(t) (A \otimes X) F(t) \, \vert \, I \otimes \mathscr{Z}_t ] (I
\otimes U(t)).
\label{eq:varsigma-def}
\end{eqnarray}
Then
\begin{equation}
\pi_t(A \otimes X) = \frac{\varsigma_t(A \otimes X) }{\varsigma_t( I \otimes
I)} .  
\label{eq:matrix-bayes}
\end{equation}
\end{lemma}

In order to determine the differential equation for $\varsigma_t(A\otimes X)$,  we first define, for $ A\otimes X\in
\mathcal{B}(\mathbb{C}^{2}\otimes \mathfrak{h}_{S})$  the process
\begin{equation}
\gamma _{t}\left( A\otimes X \right) =\mathbb{E}_{\Sigma }[F(t)^{\ast }A\otimes X \,
F(t)|I\otimes \mathscr{Z}_t],
\end{equation}
so that $\varsigma _{t}(A\otimes X)\equiv (I\otimes U(t))^{\ast }\gamma
_{t}(A\otimes X)(I\otimes U(t))$. We  then have

\begin{lemma}
\label{lemma:matrix-ref-c-exp} Let $\alpha_0 \neq 0$. The process $\gamma _{t}\left( A\otimes X \right) $
satisfies the QSDE
\begin{equation}
d\gamma _{t}\left( A\otimes X \right) =\tau _{t}\left( A\otimes X \right) dt+\beta _{t}\left(
A\otimes X \right) dZ\left( t\right)   \label{eq:qsde-gamma}
\end{equation}
where $\tau _{t}\left( A \otimes X\right) ,\beta _{t}\left( A \otimes X\right) \in \mathscr{Z}_t $, are given by
\begin{eqnarray*}
\beta _{t}\left( A \otimes X\right)  &=&\gamma _{t}\left( Q_{t}\left( A \otimes X\right) \right)
\nonumber \\ &&
+\nu \xi (t)\gamma _{t}\left( A \otimes X\left( \sigma _{+}\otimes I\right) \right)
+\nu ^{\ast }\xi (t)^{\ast }\gamma _{t}\left( \left( \sigma _{-}\otimes
I\right) A \otimes X\right) -\gamma _{t}\left( A \otimes X\right) \theta _{t}, \\
\tau _{t}\left( A \otimes X\right)  &=&\gamma _{t}\left( T_{t}\left( A \otimes X\right) \right)
\nonumber \\ &&
+\nu \xi (t)\gamma _{t}\left( Q_{t}\left( A \otimes X\right) \sigma _{+}\right) +\nu
^{\ast }\xi (t)^{\ast }\gamma _{t}\left( \sigma _{-}Q_{t}\left( A \otimes X\right)
\right) -\beta _{t}\left( A \otimes X\right) \theta _{t},
\end{eqnarray*}
with
\begin{equation*}
\theta _{t}=\mathbb{E}_{\psi }[\left( \nu \xi (t)\sigma _{+}+\nu ^{\ast }\xi
^{\ast }(t)\sigma _{-}\right) \otimes I|I\otimes \mathscr{Z}_{t}].
\end{equation*}
\end{lemma}

\begin{proof}
Setting $R_{t}=F(t)^{\ast }(A\otimes X)F(t)$, we have that
\begin{equation*}
dR_{t}=F(t)^{\ast }T_{t}\left( A\otimes X \right) F(t)dt+F(t)^{\ast
}Q_{t}(A\otimes X)F(t)dZ\left( t\right)
\end{equation*}
and our aim is compute $\gamma _{t}\left( A\otimes X \right) =\mathbb{E}_{\Sigma
}[R_{t}|I\otimes \mathscr{Z}_{t}]$. In particular,
\begin{equation}
\mathbb{E}_{\Sigma }\left[ \left( R_{t}-\gamma _{t}\left( A\otimes X \right) \right)
D_{t}\right] =0
\end{equation}
for every $D_{t}\in I\otimes \mathscr{Z}_{t}$ and we now apply a technique
similar to the characteristic function method, this time using the input
process $Z$ and taking the process $D_{t}$ to satisfy the QSDE $%
dD_{t}=f\left( t\right) D_{t}dZ\left( t\right) $ with $D_{0}=I$, for given
integrable $f$. From the Ito product rule we then have
\begin{equation}
0=\mathbb{E}_{\Sigma}\left[ \left( dR_{t}-d\gamma _{t}\left( A\otimes X \right) \right)
D_{t}+\left( R_{t}-\gamma _{t}\left( A\otimes X \right) \right) dD_{t}+\left(
dR_{t}-d\gamma _{t}\left( A\otimes X \right) \right) dD_{t}\right]
\end{equation}
and making the ansatz that $d\gamma _{t}\left( A\otimes X \right) =\tau _{t}\left(
A\otimes X \right) dt+\beta _{t}\left( A\otimes X \right) dZ\left( t\right) $ for unknown
coefficients $\tau _{t}\left( A\otimes X \right) $ and $\beta _{t}\left( A\otimes X \right) $ we
see that
\begin{eqnarray*}
0 &=&\mathbb{E}_{\Sigma } [ ( F(t)^{\ast }T_{t}\left( A\otimes X \right)
F(t)-\tau _{t}\left( A\otimes X \right) ) D_{t}dt
\nonumber \\ && 
\ \ \ \ \ + ( F(t)^{\ast }Q_{t}\left(
A\otimes X \right) F(t)-\beta _{t}\left( A\otimes X \right) ) D_{t}dZ\left( t\right) %
]  \\
&&+\mathbb{E}_{\psi }\left[ \left( R_{t}-\gamma _{t}\left( A\otimes X \right)
\right) D_{t}f\left( t\right) dZ(t)\right]  \\
&&+\mathbb{E}_{\psi }\left[ \left[ F(t)^{\ast }Q_{t}\left( A\otimes X \right)
F(t)-\beta _{t}\left( A\otimes X \right) \right] D_{t}f\left( t\right) dt\left(
t\right) \right] .
\end{eqnarray*}
We now make use of Lemma \ref{lemma:expectation-basic-extended} again and
apply the commutation relations $F(t)\sigma _{+}=\sigma _{+}F(t)$, $\sigma
_{-}F(t)^{\ast }=F^{\ast }\left( t\right) \sigma _{-}$. (Note that $\sigma
_{+}$ will not commute with $F^{\ast }\left( t\right) $.) Inserting $\mathbb{%
E}_{\Sigma}\left[ \cdot |I\otimes \mathscr{Z}_t \right] $ under the
expectation sign, then separating coefficients of $D_t$ and $D_tf(t)$, we obtain
the equations
\begin{eqnarray*}
0 &=&\gamma _{t}\left( T_{t}\left( A\otimes X \right) \right) -\tau _{t}\left(
A\otimes X \right) 
\\ &&
+\nu _{t}\gamma _{t}\left( Q_t\left( A\otimes X \right) \left( \sigma
_{+}\otimes I\right) \right) 
+\nu _{t}\gamma _{t}\left( \left( \sigma
_{-}\otimes I\right) Q_t\left( A\otimes X \right) \right)  
\\
&&
-\beta _{t}\left( A\otimes X \right) \mathbb{E}_{\Sigma }\left[ \nu \xi (t)\sigma
_{+}\otimes I+\nu ^{\ast }\xi (t)^{\ast }\sigma _{-}\otimes I|I\otimes
\mathscr{Z}_t  \right] , \\
0 &=&\gamma _{t}\left( Q_t\left( A\otimes X \right) \right) -\beta _{t}\left( A\otimes X \right)
\\ &&
+\nu \xi (t)\gamma _{t}\left( A\otimes X \left( \sigma _{+}\otimes I\right) \right)
+\nu ^{\ast }\xi (t)^{\ast }\gamma _{t}\left( \left( \sigma _{-}\otimes
I\right) A\otimes X \right)  \\
&& -\gamma _{t}\left( A\otimes X \right) \mathbb{E}_{\Sigma }\left[ \nu \xi
(t)\sigma_{+}\otimes I+\nu ^{\ast }\xi (t)^{\ast } \sigma _{-}\otimes
I|I\otimes \mathscr{Z}_t \right] .
\end{eqnarray*}
Rearranging these expressions yields the relations in the statement of the
lemma.
\end{proof}

\begin{theorem}
Let $\alpha_0 \neq 0$. The unnormalized conditional expectation $\varsigma_t(A \otimes X)$ defined by (\ref{eq:varsigma-def}) 
satisfies the
 equation
\begin{equation}
d\varsigma _{t}\left(  A \otimes X \right) =\varsigma _{t}\left( \mathcal{G}%
_{t}(A \otimes X)\right) dt+\lambda _{t}\left( A \otimes X \right) d\tilde{Y}(t), \label{eq:matrix-sigma-dyn}
\end{equation}
where
\begin{eqnarray}
\lambda _{t}\left( A \otimes X \right)  &=&\varsigma _{t}(A \otimes X  \tilde{L}_{t}+\tilde{L}%
_{t}^{\ast }A \otimes X)-\varsigma _{t}\left( A \otimes X \right) \kappa _{t}, \label{eq:def_lambda_t}\\
\tilde{L}_{t} &=&I\otimes L+\nu _{t}\xi (t)\sigma _{+}\otimes S,
\label{Ltilde} \\
d\tilde{Y}(t) &=&dY\left( t\right) -\kappa _{t}dt,\quad \tilde{Y}(0)=0,
\label{Ytilde} \\
\kappa _{t} &=&\varsigma _{t}\left( \left( \nu \xi (t)\sigma _{+}+\nu ^{\ast
}\xi (t)^{\ast }\sigma _{-}\right) \otimes I\right) .
\label{eq:kappa-def}
\end{eqnarray}
\end{theorem}

\begin{proof}
We remark that by inspection the coefficients in the QSDE for $\gamma
_{t}\left( A \otimes X \right) $ simplify to
\begin{eqnarray*}
\beta _{t}\left(  A \otimes X \right)  &=&\gamma _{t}\left( Q_{t}\left(A \otimes X \right) \right)
\\
&& +\nu \xi (t)\gamma _{t}\left( A \otimes X \left( \sigma _{+}\otimes I\right) \right)
+\nu ^{\ast }\xi (t)^{\ast }\gamma _{t}\left( \left( \sigma _{-}\otimes
I\right) A \otimes X \right) -\gamma _{t}\left( A \otimes X \right) \theta _{t} \\
&=&\gamma _{t} ( A \otimes X \left( I\otimes L+\nu _{t}\sigma _{+}\otimes S\right)
+\left( I\otimes L^{\ast }+\nu ^{\ast }\xi (t)^{\ast }\sigma _{-}\otimes
S^*\right) A \otimes X )
\\ &&
 -\gamma _{t}\left( A \otimes X \right) \theta _{t}, 
\\
\tau _{t}\left( A \otimes X \right)  &=&\gamma _{t}\left( T_{t}\left( A \otimes X \right) \right)
+\nu \xi (t)\gamma _{t}\left( Q_{t}\left( A \otimes X \right) \sigma _{+}\right) 
\\
&&
+\nu
^{\ast }\xi (t)^{\ast }\gamma _{t}\left( \sigma _{-}Q_{t}\left( A \otimes X \right)
\right) -\beta _{t}\left( A \otimes X \right) \theta _{t} \\
&\equiv &\gamma _{t}\left( \mathcal{G}\left( A \otimes X \right) \right) -\beta
_{t}\left( A \otimes X \right) \theta _{t},
\end{eqnarray*}
and therefore
\begin{equation*}
d\gamma _{t}\left( A \otimes X \right) =\gamma _{t}\left( \mathcal{G}(A \otimes X )\right)
dt+\beta _{t}\left( A \otimes X \right) \left[ dZ\left( t\right) -\theta _{t}dt\right] .
\end{equation*}
The QSDE for $\varsigma _{t}(A \otimes X)\equiv (I\otimes U(t))^{\ast }\gamma
_{t}(A \otimes X)(I\otimes U(t))$ is then readily deduced from the unitary rotation
noting that $\kappa _{t}\equiv (I\otimes U(t))^{\ast }\theta _{t}(I\otimes
U(t))$ and $\lambda _{t}\left( A \otimes X \right) \equiv (I\otimes U(t))^{\ast }\beta
_{t}\left( A \otimes X \right) (I\otimes U(t))$.
\end{proof}

\begin{corollary}
For $\alpha_0 \neq 0$, the conditional expectation $\pi_t(A \otimes X)$  defined by (\ref{eq:matrix-c-exp-def}) and given by (\ref{eq:matrix-bayes}) 
satisfies   equation (\ref{pi_filter}) derived in Theorem \ref{thm:matrix-filter}.
\end{corollary}

\begin{proof}
We see that $d\varsigma _{t}\left( I\otimes I\right) =\lambda _{t}\left(
I\otimes I\right) \,d\tilde{Y}\left( t\right) $ and so
\begin{equation*}
d\frac{1}{\varsigma _{t}\left( I\otimes I\right) }=-\frac{\lambda _{t}\left(
I\otimes I\right) }{\varsigma _{t}\left( I\otimes I\right) ^{2}}\,d\tilde{Y}%
\left( t\right) +\frac{1}{\varsigma _{t}\left( I\otimes I\right) ^{3}}%
\lambda _{t}\left( I\otimes I\right) ^{2}dt.
\end{equation*}
However, we note from (\ref{eq:def_lambda_t}) that
\begin{equation*}
\frac{\lambda _{t}\left( I\otimes I\right) }{\varsigma _{t}\left( I\otimes
I\right) }\equiv \pi _{t}\left( \tilde{L}_{t}+\tilde{L}_{t}^{\ast }\right)
-\kappa _{t}.
\end{equation*}
By an application of the Ito product rule, the normalized filter therefore
satisfies
\begin{eqnarray*}
&& d\pi _{t}(A \otimes X) 
\\
&=&\pi _{t}(\mathcal{G}_{t}(A \otimes X))dt \\
&&+\left\{ \frac{\lambda _{t}\left( A \otimes X \right) -\pi _{t}\left( A \otimes X \right)
\lambda _{t}\left( I\otimes I\right) }{\varsigma _{t}\left( I\otimes
I\right) }\right\} \left[ dY(t) -\kappa _{t}dt-\frac{\lambda _{t}\left( I\otimes
I\right) }{\varsigma _{t}\left( I\otimes I\right) }dt\right]  \\
&=&\pi _{t}(\mathcal{G}_{t}(A \otimes X))dt+ \{ \pi _{t}\left( A \otimes X \tilde{L}_{t}+%
\tilde{L}_{t}^{\ast }A \otimes X \right) 
\\
&&
-\pi _{t}\left( A \otimes X \right) \pi _{t}\left( \tilde{%
L}_{t}+\tilde{L}_{t}^{\ast }\right) \}  \left[ dY(t)-\pi _{t}\left( \tilde{L}_{t}+\tilde{L}%
_{t}^{\ast }\right) dt\right]  \\
&\equiv &\pi _{t}(\mathcal{G}_{t}(A \otimes X))dt+\mathcal{H}_{t}(A \otimes X)dW\left( t\right),
\end{eqnarray*}
since we have from (\ref{Ltilde}), (\ref{eq:matrix-filter}) and(\ref
{eq:innovation-matrix})
\begin{eqnarray*}
\mathcal{H}_{t}(A \otimes X) &\equiv &\pi _{t}\left( A \otimes X \tilde{L}_{t}+\tilde{L}%
_{t}^{\ast }A \otimes X \right) -\pi _{t}\left( A \otimes X \right) \pi _{t}\left( \tilde{L}_{t}+%
\tilde{L}_{t}^{\ast }\right) , \\
dW(t) &\equiv &dY-\pi _{t}\left( \tilde{L}_{t}+\tilde{L}_{t}^{\ast }\right)
dt.
\end{eqnarray*}
Therefore we recover equation (\ref{pi_filter}).
\end{proof}

\section{Fields in a Superposition of Coherent States}
\label{sec:cat}

\subsection{Superposition of Coherent States}
\label{sec:cat-super}

In this section we take the field to be in a superposition state
\begin{equation}
 \vert \Psi \rangle   = \sum_j \alpha_j \vert f_j \rangle ,
\label{eq:super-state}
\end{equation}
where $\vert f_j \rangle$ are coherent states and the complex numbers $\alpha_j$  ($j=1,\ldots,n$) are non-zero normalized weights (described further below).

Coherent vectors $\vert f \rangle$ may be expressed in terms of the vacuum vector using the Weyl (or displacement) operator  \cite{KRP92}  $W(f)$ which serves as a \lq\lq{density}\rq\rq:
\begin{equation}
\vert f \rangle = W(f) \vert 0 \rangle.
\end{equation}
While the collection of all coherent vectors is dense in the Fock space, they are not orthogonal, and indeed the inner product (in the Fock space) is given by
\begin{equation}
\langle f \vert  g \rangle = \exp(  -\frac{1}{2} \parallel f \parallel_2^2  -\frac{1}{2} \parallel g \parallel_2^2 + \langle f, g \rangle_2 ).
\end{equation}
Here, $\parallel f \parallel_2^2 = \langle f,f \rangle_2$ and $\langle f, g \rangle_2$ are the $L^2([0,\infty), \mathbf{C})$ norm and inner product respectively.
The superposition state $\vert \psi \rangle$ given by (\ref{eq:super-state}) is specified by a choice of coherent vectors $\vert f_j \rangle$, with weights $\alpha_j$ ensuring normalization: $\langle \psi \vert \psi \rangle = \sum_{jk} \alpha_j^\ast \alpha_k g_{jk}=1$, where $g_{jk}= \langle f_j \vert f_k \rangle$.

For a  system operator  $X$  acting on $\mathfrak{H}_S$,
and $F$ is a field operator acting on the Fock space $\mathfrak{F}$, the  expectation
with respect to the state $ \vert \eta \rangle  \otimes  \vert \Psi \rangle$
 is defined by
\begin{eqnarray}
\mathbb{E}_{\eta\Psi}[ X \otimes F] &=&  \langle \eta \Psi  \vert  (X \otimes F)\vert \eta\Psi \rangle
=
\langle \eta \vert  X \vert \eta \rangle \langle \Psi \vert  F \vert  \Psi \rangle
\nonumber \\
&=&
\langle \eta \vert X \vert \eta \rangle \sum_{jk} \alpha_j^\ast \alpha_k \langle f_j \vert  F  \vert f_k \rangle
\nonumber
\\
&=& \sum_{jk} \alpha_j^\ast \alpha_k \mathbb{E}_{jk}[ X \otimes F ],
\label{eq:expect-def}
\end{eqnarray}
where
\begin{equation}
\mathbb{E}_{jk}[ X \otimes F ] = \langle \eta \vert  X \vert \eta \rangle \langle  f_j \vert  F  \vert  f_k \rangle 
\end{equation}
for $j,k=1,\ldots,n$. We write $\mathbb{E}_{00}[ X \otimes F ] = \langle \eta \vert  X \vert \eta \rangle \langle  0 \vert  F \vert  0 \rangle $ for the vacuum case.

Consider now the expectation of an adapted operator $K(t)$ on the composite system $\mathfrak{H} = (\mathfrak{H}_S \otimes \mathfrak{F}_{t]}) \otimes \mathfrak{F}_{(t}$; this means that $K(t)$ acts trivially on the future component $\mathfrak{F}_{(t}$.
Let  $\chi_{[0,t]}$ is the indicator function for the time interval $[0,t]$. Now coherent vectors and Weyl operators factorize as $\vert f \rangle = \vert f \chi_{[0,t]} \rangle \otimes \vert f \chi_{(t,\infty)} \rangle$ and $W(f)= W(f \chi_{[0,t]} )  \otimes W(f \chi_{(t,\infty)} ) $, respectively. 
Write
\begin{equation}
W^-_t(f)= W(f \chi_{[0,t]} ), \ \ W^+_t(f)=W(f \chi_{(t,\infty)} ) .
\end{equation}
Then we can express the coherent expectations of adapted processes $K(t)$  in terms of the vacuum:
\begin{eqnarray}
\mathbb{E}_{jk}[ K(t)] &=& \mathbb{E}_{00} [ W^{-\ast}_t(f_j) K(t) W^-_t(f_k) ] r^{jk}(t)
\label{eq:Kt-00-0}
\end{eqnarray}
where
$
r^{jk}(t) = \langle 0 \vert W^{+\ast}_t(f_j) W^+_t(f_k) 
\vert 0 \rangle 
$
satisfies
\begin{equation}
\dot r^{jk}(t) = -( f_j^\ast (t)f_k(t) - \frac{1}{2} \vert f_j(t) \vert^2 - \frac{1}{2} \vert f_k (t) \vert^2 )r^{jk}(t), \ \ r^{jk}(0)=1.
\end{equation}
Note that $j=k$ is the standard coherent expectation, in which case $r^{jj}(t)=1$.

The following lemma shows how expectations of stochastic integrals with respect to the superposition state can be evaluated.

\begin{lemma}
\label{lemma:expectation-basic-cat} 
Let $K(t)$ be a bounded quantum stochastic process 
defined by (\ref{eq:Kt-def}), 
where $M_0$, $M_\pm$ and $M_1$ are bounded and adapted. Then we have
\begin{eqnarray}
\mathbb{E}_{jk}[ K(t) ] &=& \mathbb{E}_{jk}[ \int_0^t M_0(s) ds ] +  \int_0^t M_- (s) f_k(s) ds 
\nonumber \\
&& 
+  \int_0^t M_+ (s) f^\ast_j(s) ds + \int_0^t M_1 (s) f_j^\ast(s) f_k(s) ds
]  .
\label{eq:Kt-00-cat}
\end{eqnarray}
\end{lemma}

\begin{proof}
Equation (\ref{eq:Kt-00-cat}) follows from the following eigenstate property of coherent vectors:
\begin{eqnarray}
dB(t) \vert f \rangle &=&  f(t) \vert f \rangle dt, 
\notag \\
d\Lambda(t) \vert f \rangle &=&  dB^\ast(t) f(t) \vert f \rangle .
\end{eqnarray}
\qed
\end{proof}

\subsection{Embedding}
\label{sec:cat-embed}

For the superposition of $n$ coherent states, we use an $n$-level ancilla system, leading to
 the  extended space
\begin{equation}
\tilde {\mathfrak{H}}= \mathbb{C}^n \otimes \mathfrak{H} = \mathfrak{H}
\oplus \mathfrak{H} \oplus \cdots \oplus \mathfrak{H} \ \ (n \ \mathrm{times}).
\end{equation}

As in the single photon case, we allow the extended system to evolve unitarily according to $I \otimes U(t)
$, where $U(t)$ is the unitary operator for the system and field, given by
the Schr\"{o}dinger equation (\ref{eq:unitary}). 
Let $\vert e_j \rangle$, $j=1,\ldots,n$, be an orthonormal basis for $\mathbf{C}^n$.
We initialize the extended system in
the   state
\begin{equation}
 \vert \Sigma \rangle  =  \frac{1}{\vert \alpha \vert} \sum_{j}  \alpha_j   \vert e_j \rangle \otimes \vert \eta \rangle  \otimes \vert f_j \rangle, 
  \label{eq:super-cat}
\end{equation}
where $\alpha_j \neq 0$ for all $j$ and $\vert \alpha \vert^2= \sum_j \alpha_j^\ast \alpha_j$ (so that $\langle \Sigma \vert \Sigma \rangle =1$).
This state evolves according to $\vert \Sigma(t) \rangle = (I \otimes U(t))  \vert \Sigma \rangle$.

Let $A$ be an operator acting on $\mathbb{C}^n$, i.e. a complex $n
\times n$ matrix, $A=( a_{jk} )$, $j,k=1,\ldots,n$. Then expectation in the extended system is defined by
\begin{equation}
\mathbb{E}_{\Sigma}[ A \otimes X \otimes F] = \langle \Sigma \vert (A \otimes X \otimes F) \vert \Sigma \rangle
= \frac{1}{\vert \alpha \vert^2} \sum_{jk} a_{jk} \alpha_j^\ast \alpha_k \mathbb{E}_{jk}[ X \otimes F].
\label{eq:expect-extended-def}
\end{equation}

Expectations of quantum stochastic integrals can be compactly expressed in the extended system, as the following lemma shows.

\begin{lemma}
\label{lemma:expectation-basic-extended-cat} 
Let $M(t)$ be adapted. Then
\begin{eqnarray}
\mathbb{E}_{\Sigma} [ \int_0^t A \otimes M(s) dB(s) ] &=&  
\mathbb{E}_{\Sigma}  [\int_0^t ( A C(s) )\otimes M(s)  ds ], 
\label{eq:extended-expect-1}
\\
\mathbb{E}_{\Sigma}  [ \int_0^t A \otimes M (s) dB^\ast(s) ] &=& 
 \mathbb{E}_{\Sigma}  [ \int_0^t ( C^\dagger (s) A )\otimes M(s) ds], 
 \label{eq:extended-expect-2}
 \\
\mathbb{E}_{\Sigma}  [ \int_0^t A \otimes M(s) d\Lambda(s) ] &=& 
 \mathbb{E}_{\Sigma}  [ \int_0^t ( C^\dagger(s) A C(s)  )\otimes M(s) ds ],
 \label{eq:extended-expect-3}
\end{eqnarray}
where
\begin{equation}
C(t) = \mathrm{diag}[ f_1(t), \ldots, f_n(t) ].
\end{equation}
\end{lemma}

Notice that the expectations  of the stochastic integrals  are expressed in terms of the action of the matrix $C(t)$ on the ancilla factor $A$.

\subsection{Master Equation}
\label{sec:cat-master}

In this section we show how the
the unconditional expectation
\begin{equation}
\mu_t(X) = \mathbb{E}_{\eta\Psi}[ X(t) ]
\label{eq:mu-def}
\end{equation}
may be computed from a collection of differential equations. We  do this through a differential equation for the unconditional expectation
\begin{equation}
\tilde\mu_t(A\otimes X) = \mathbb{E}_{\Sigma} [ A \otimes X(t) ]
\label{eq:mu-tilde-def}
\end{equation}
for the extended system. 

Let $R$ be an $n \times n$ matrix defined by $R_{jk}=1$ for all $j,k = 1, \ldots, n$, and
define
\begin{eqnarray}
\mathcal{G}_t(A \otimes X) &=& A \otimes \mathcal{L}(X) + ( AC(t))
\otimes [L^\ast, X] S   + (C^\dagger(t) A ) \otimes S^\ast [X,L] 
 \notag \\
&& + ( C^\dagger(t) A C(t) ) \otimes (S^\ast X S - X)) .
\end{eqnarray}

\begin{lemma}  \label{lemma:master-1-cat}
The unconditional expectation (\ref{eq:mu-def}) with respect to the superposition state $\vert \Psi \rangle$ (defined by (\ref{eq:super-state})) is given by
\begin{equation}
\mu_t(X)= \frac{ \tilde\mu_t(R \otimes X) }{ \tilde\mu_t( R \otimes I) },
\label{eq:mu-relation}
\end{equation}
and the master equation for the expectation (\ref{eq:mu-tilde-def}) in the 
extended system is
\begin{equation}
\frac{d}{dt} \tilde \mu_t(A \otimes X) = \tilde \mu_t( \mathcal{G}_t(A \otimes X)) ,
\label{eq:mu-master}
\end{equation}
with initial condition $\tilde \mu_0(A\otimes X)=  \frac{1}{\vert \alpha \vert^2}    \langle \eta \vert  X \vert \eta \rangle  \sum_{jk} a_{jk} \alpha_j^\ast \alpha_k$.
\end{lemma}

\begin{proof} By definitions (\ref{eq:expect-extended-def})  and (\ref{eq:expect-def}) we have
\begin{eqnarray}
\mathbb{E}_{\Sigma} [  R\otimes X(t) ] &=&    \frac{1}{\vert \alpha \vert^2} \sum_{jk}  \alpha_j^\ast \alpha_k \mathbb{E}_{jk}[ X(t)  ]
\\
&=&
 \frac{1}{\vert \alpha \vert^2}  \mathbb{E}_{\eta\Psi} [   X(t) ] ,
\end{eqnarray}
and in particular
\begin{equation}
\mathbb{E}_{\Sigma} [  R \otimes I ] =  \frac{1}{\vert \alpha \vert^2}   .
\end{equation}
From these expressions, we see that
\begin{equation}
\mathbb{E}_{\eta\Psi} [   X(t) ]  = \vert \alpha \vert^2 \mathbb{E}_{\Sigma} [  R\otimes X(t) ]  = \frac{\mathbb{E}_{\Sigma} [  R\otimes X(t) ]  }{\mathbb{E}_{\Sigma} [  R\otimes I] },
\end{equation}
which proves (\ref{eq:mu-relation}).

The differential equation (\ref{eq:mu-master}) follows from the QSDE (\ref{eq:X-dyn}) for $X(t)=j_t(X)$ and relations (\ref{eq:extended-expect-1})-(\ref{eq:extended-expect-3}) upon evaluating the differential  $d \mathbb{E}_{\Sigma}[ A \otimes X(t) ]$.
\qed
\end{proof}

\begin{theorem}
\label{thm:master-cat} 
The unconditional expectation $\mu_t(X)$
when the field is in the superposition state  $\vert \Psi \rangle$ (defined by (\ref{eq:super-state}))  is given by
\begin{equation}
\mu_t(X) = \frac{ \sum_{jk} \alpha_j^\ast \alpha_k   \mu^{jk}_t(X) }{\sum_{jk} \alpha_j^\ast \alpha_k   \mu^{jk}_t(I) },
\label{eq:mu-cat-rep}
\end{equation}
where $\mu^{jk}_t(X)$ is given by
the system of equations
\begin{equation}
\frac{d}{dt} \mu^{jk}_t(X) = \mu^{jk}_t( \mathcal{G}^{jk}_t(X) ),
\label{eq:dot-mu-jk-c}
\end{equation}
and where
\begin{equation}
 \mathcal{G}^{jk}_t(X)  = \mathcal{L}(X) + S^\ast [X, L ] f_j^\ast(t) + [ L^\ast, X] S f_k(t) + (S^\ast X S - X) f_j^\ast(t) f_k(t) .
\end{equation}
The initial conditions are
\begin{equation}
\mu^{jk}_0(X)= \langle \eta \vert X \vert \eta \rangle g_{jk}.
\end{equation}

\end{theorem}

\begin{proof}
Define
\begin{equation}
\mu^{jk}_t(X) = \mathbb{E}_{jk}[ X(t) ].
\end{equation}
Then as in the proof of Lemma \ref{lemma:master-1-cat} we may show that
\begin{equation}
\mu^{jk}_t(X) = \frac{\vert \alpha \vert^2}{\alpha_j^\ast \alpha_k}  \tilde\mu_t( \vert e_j \rangle \langle e_k \vert   \otimes X) .
\end{equation}
The the relation (\ref{eq:mu-cat-rep}) follows from (\ref{eq:mu-relation}). The differential equation (\ref{eq:dot-mu-jk-c}) follows from equation (\ref{eq:mu-master}) with $A=\vert e_j \rangle \langle e_k \vert$.
\qed
\end{proof}

\subsection{Superposition State Filter}
\label{sec:cat-super-filter}

In this section we show how the conditional expectation $\hat X(t)=\pi_t(X)$ defined by (\ref{eq:cond-exp})
can be evaluated using a system of conditional equations. This will make use of the conditional expectation
\begin{equation}
\tilde \pi_t( A \otimes X) = \mathbb{E}_{\Sigma} [ A \otimes X(t) \, \vert \, I \otimes \mathscr{Y}_t].
\end{equation}
for the extended system.

\begin{lemma}     \label{lemma:filter-extended-cat}
The  conditional expectation  $\hat X(t)=\pi_t(X)$ defined by (\ref{eq:cond-exp}) with respect to the superposition state $\vert \Psi\rangle $ is given by
\begin{equation}
\pi_t(X)= \frac{ \tilde\pi_t(R \otimes X) }{ \tilde\pi_t( R \otimes I) }.
\label{eq:pi-relation}
\end{equation}
The quantum filter for the conditional expectation $\tilde\pi_t(A\otimes X)$ is
\begin{eqnarray}
d \tilde \pi_t( A \otimes X) & = & \tilde \pi_t( \mathcal{G}_t(A \otimes X)) dt + \mathcal{H}_t( A \otimes X) dW(t)
\label{eq:super-filter}
\end{eqnarray}
with initial condition $\tilde\pi_0(A\otimes X)=  \frac{1}{\vert \alpha \vert^2}    \langle \eta, X \eta \rangle  \sum_{jk} a_{jk} \alpha_j^\ast \alpha_k$,
where
\begin{eqnarray}
\mathcal{H}_t( A \otimes X)  &=& \tilde \pi_t(  A \otimes X (I \otimes L + C(t) \otimes S) 
+ (I \otimes L^\ast + C^\dagger(t) \otimes S^\ast) A \otimes X
)
\nonumber \\
&&
- \tilde \pi_t( A \otimes X) \tilde  \pi_t( I \otimes L + C(t) \otimes S
+
 I \otimes L^\ast + C^\dagger(t) \otimes S^\ast
)
\end{eqnarray}
and $W(t)$ is a $\mathscr{Y}_t$-Wiener process given by
\begin{eqnarray}
dW(t) &=&
dY(t) - \tilde \pi_t(   I \otimes L + C(t) \otimes S
+
 I \otimes L^\ast + C^\dagger(t) \otimes S^\ast ) dt, \ W(0)=0.
\end{eqnarray}
The initial condition is
$\tilde \pi_0(A\otimes X)=  \frac{1}{\vert \alpha \vert^2}    \langle \eta \vert  X \vert \eta \rangle  \sum_{jk} a_{jk} \alpha_j^\ast \alpha_k$.
\end{lemma}

\begin{proof}
Let $K \in \mathscr{Y}_t$. Then we have
\begin{eqnarray}
\mathbb{E}_{\Sigma} [ \tilde\pi_t( R\otimes X) (I\otimes K) ] &=&  \mathbb{E}_{\Sigma} [ (R \otimes X(t)) (I \otimes K) ]
\notag
\\
&=&  \frac{1}{\vert \alpha \vert^2} \sum_{jk}  \alpha_j^\ast \alpha_k \mathbb{E}_{jk}[ X(t) K]
\notag
\\
&=&
 \frac{1}{\vert \alpha \vert^2}  \mathbb{E}_{\eta\Psi}[ X(t)K ]
 \notag
 \\
 &=&
 \frac{1}{\vert \alpha \vert^2}  \mathbb{E}_{\eta\Psi }[ \pi_t(X) K ]
 \notag
 \\
 &=&
 \mathbb{E}_{\Sigma}[ R \otimes \pi_t(X) K ]
  \notag
 \\
 &=&
 \mathbb{E}_{\Sigma}[  \mathbb{E}_{\Sigma} [R \otimes \pi_t(X) K \vert I \otimes \mathscr{Y}_t ]]
 \notag \\
 &=&
 \mathbb{E}_{\Sigma}[ \tilde\pi_t(R\otimes I) (I \otimes \pi_t(X)) (I \otimes K) ] .
\end{eqnarray}
This proves (\ref{eq:pi-relation}).

The filtering equation (\ref{eq:super-filter}) is derived 
using
the characteristic function method \cite{HSM05}, \cite{VPB93}, \cite{VPB92}.
We postulate that the filter has
the form
\begin{equation}
d \tilde\pi_t( A \otimes X) = \mathcal{F}_t(A \otimes X) dt + \mathcal{H}_t(A
\otimes X) I\otimes dY(t) ,
\end{equation}
where $\mathcal{F}_t$ and $\mathcal{H}_t$ are to be determined.

Let $f$ be square integrable, and define a process $c_f(t)$ by
\begin{equation}
dc_f(t) = f(t) c_f(t) dY(t), \ \ c_f(0)=1.
\end{equation}
Then $I \otimes c_f(t)$ is adapted to $I \otimes \mathscr{Y}_t$, and the
definition of quantum conditional expectation \cite[sec. 3.3]{BHJ07}
 implies that
\begin{equation}
\mathbb{E}_\Sigma [ A \otimes (X(t) c_f(t) ) ] = \mathbb{E}_\Sigma [ \tilde \pi_t( A
\otimes X(t)) (I \otimes c_f(t) ) ]
\end{equation}
holds for all $f$. By calculating the differentials of both sides, taking
expectations and conditioning we obtain
\begin{eqnarray}
\mathbb{E}_\Sigma [ A \otimes (dX(t) c_f(t)) ] &=& \mathbb{E}_\Sigma [ (I
\otimes c_f(t) ) \tilde\pi_t( \mathcal{G}(A \otimes X) ) 
\\
&& + (I \otimes f(t) c_f(t) )   \tilde \pi_t( (A \otimes X)(I \otimes L + C(t) \otimes S) 
\notag
\\ && + 
(I \otimes L + C^\dagger(t) \otimes S^\ast) (A \otimes X)
 )]dt    \ \ \ 
\notag
\end{eqnarray}
and
\begin{eqnarray}
&& \mathbb{E}_\Sigma [ A \otimes (d\tilde \pi_t(A\otimes X) c_f(t)) ] 
\\
 &=& \mathbb{E}%
_\Sigma [ (I \otimes c_f(t) \{ \mathcal{F}_t(A\otimes X) + \mathcal{H}_t(A\otimes X)  
\tilde \pi_t(
I\otimes L+ C(t)\otimes S  + I\otimes L^\ast + C^\dagger(t) \otimes S^\ast
) \}
 \notag \\
&& + (I \otimes f(t) c_f(t) ) \{ \ 
\tilde \pi_t(A\otimes X) \tilde \pi_t( I\otimes L+ C(t)\otimes S  + I\otimes L^\ast + C^\dagger(t) \otimes S^\ast)+ \mathcal{H}_t(A\otimes X)  
 \} ] dt.  \notag
\end{eqnarray}
Now equating coefficients of $c_f(t)$ and $f(t) c_f(t)$ we solve for $%
\mathcal{F}_t(A\otimes X)$ and $\mathcal{H}_t(A\otimes X)$ to obtain the filtering  equation (\ref{eq:super-filter}).

We now show that $W(t)$ is a $\mathscr{Y}_t$-martingale, and since $dW(t)dW(t)=dt$, then by Levy's theorem \cite{RE82} we have that $W(t)$ is a $\mathscr{Y}_t$-Wiener process. Indeed,  for any $K \in \mathscr{Y}_t$ we have
\begin{eqnarray}
&& \mathbb{E}_{\Sigma}[  (I\otimes dW(t) )(I\otimes K) ] 
\notag \\
&=& \mathbb{E}_{\Sigma}[  (I\otimes dY(t)    - \tilde\pi_t(I\otimes L+ C(t)\otimes S  + I\otimes L^\ast + C^\dagger(t) \otimes S^\ast) dt)(I\otimes K) ]
\notag
\\
&=&
\mathbb{E}_{\Sigma}[  I\otimes (L(t)+L^\ast(t))   + C(t) \otimes S + C^\ast(t) \otimes S^\ast 
\notag
\\
&&
- \tilde\pi_t( I\otimes L+ C(t)\otimes S  + I\otimes L^\ast + C^\dagger(t) \otimes S^\ast) dt)(I\otimes K) ] dt =0.
\end{eqnarray}
This completes the proof.
\qed
\end{proof}

\begin{theorem}
\label{thm:filter-cat} 
The unconditional expectation $\mu_t(X)$
when the field is in the superposition state  $\vert \Psi \rangle$ (defined by (\ref{eq:super-state}))  is given by
\begin{equation}
\pi_t(X) = \frac{ \sum_{jk} \alpha_j^\ast \alpha_k   \pi^{jk}_t(X) }{\sum_{jk} \alpha_j^\ast \alpha_k   \pi^{jk}_t(I) },
\label{eq:pi-cat-rep}
\end{equation}
where the conditional quantities $\pi^{jk}_t(X)$ are given by
\begin{eqnarray}
d \pi^{jk}_t(X) &=& \pi^{jk}_t( \mathcal{G}^{jk}(X) )dt
+ (\pi^{jk}_t( X(L +Sf_k(t) ) + (L^\ast + S^\ast f_j^\ast(t) ) X) 
\notag
\\
&&
- \pi^{jk}_t(X) \sum_j    \frac{\vert \alpha_j \vert^2}{\vert \alpha \vert^2} \pi^{jj}_t( L+ S f_j(t) + L^\ast + S^\ast f_j^\ast(t) ) ) dW(t)
\notag
\end{eqnarray}
The innovations process $W(t)$ is a $\mathscr{Y}_t$ Wiener process with
respect to the superposition  state $\vert \Psi \rangle$ and is given by
\begin{equation}
dW(t) = dY(t) - \sum_j    \frac{\vert \alpha_j \vert^2}{\vert \alpha \vert^2} \pi^{jj}_t( L+ S f_j(t) + L^\ast + S^\ast f_j^\ast(t) ) dt .
 \label{eq:innovation-11-cat}
\end{equation}
The initial conditions are
\begin{equation}
\pi^{jk}_0(X)= \langle \eta \vert X \vert \eta \rangle g_{jk}.
\end{equation}
\end{theorem}

\begin{proof}
These assertions follow upon substitution of
\begin{equation}
\pi^{jk}_t(X) = \frac{\vert \alpha \vert^2}{\alpha_j^\ast \alpha_k}  \tilde\pi_t( e_je_k^\ast  \otimes X) .
\end{equation}
into the relevant expressions from Lemma \ref{lemma:filter-extended-cat}.
\qed
\end{proof}

\section{Discussion and Conclusion}
\label{sec:conclusion}

In this paper we have derived the master equation and quantum filter  for a class of open quantum systems
that are coupled to continuous-mode fields in non-classical states: (i) single photon states, and (ii) superpositions of coherent states.
The quantum  filter  in both of the cases we consider  consists of   coupled equations  that determine the evolution of the conditional state of the system under continuous (weak) measurement performed on the output field, in contrast to 
the familiar single filtering equation for open Markov quantum systems that are coupled to coherent boson fields. 
 This coupled equations structure of the master and filter equations is a reflection of the non-Markov nature of systems coupled to the non-classical fields. 
  Indeed, a key feature of our approach is the  embedding of the system into a larger extended system, a technique often employed in the analysis of non-Markov  systems, providing an elegant framework within which to study the the dynamics, both unconditional and conditional, of the system.  In contrast to Markovian embeddings  \cite{HPB04}, \cite{GJN11a}, \cite{GJNC11}, the extended system (including the field) is initialized in a superposition state. This embedding provides a framework within which the tools of the quantum stochastic calculus may be efficiently applied to determine quantum filtering equations. We expect that the use of suitable embeddings, both Markovian and non-Markovian,  could be  adapted to study quantum systems that are coupled to other types of highly non-classical  fields.

\section*{Acknowledgement}
The authors wish to thank J.~Hope for helpful discussions and for pointing out reference \cite{HPB04} to us. We also wish to thank A.~Doherty, H.~Wiseman, E.~Huntington and J.~Combes for help discussions and suggestions.



\bibliographystyle{siam}

\end{document}